\documentclass[a4paper]{article}
\usepackage{a4wide, authblk}
\usepackage{amsmath,mathtools}
\usepackage{amsfonts}
\usepackage{amssymb}
\usepackage{amsthm}
\usepackage{bbm}
\usepackage{color}
\usepackage{graphicx}
\usepackage{mathrsfs}
\usepackage{enumitem}
\usepackage{xspace}
\usepackage{booktabs}
\usepackage[english]{babel}
\usepackage{marvosym}
\usepackage{hyperref} 
\usepackage{multirow,tabularx}
\usepackage[utf8]{inputenc}
\usepackage[mathscr]{eucal}
\usepackage{bm}

\usepackage{tikz}
\usepackage{subcaption}

\newcommand{\prob}{\mathbb{P}}
\newcommand{\Prob}[1]{\prob\left(#1\right)}

\newcommand{\expec}{\mathbb{E}}
\newcommand{\Exp}[1]{\expec\left[#1\right]}

\newcommand{\Var}[1]{\textup{Var}\left(#1\right)}

\newcommand{\ind}[1]{\mathbbm{1}_{\left\{#1\right\}}}

\newcommand\abs[1]{\left|#1\right|}

\newcommand{\dd}{{\rm d}}

\newcommand{\thetao}[1]{\Theta\left(#1\right)}

\definecolor{color1vertex}{HTML}{fff0f5}%
\definecolor{color2vertex}{rgb}{0.6314,0.8549,0.7059}%
\definecolor{colorsqrt}{HTML}{ff3799}%
\definecolor{colornm}{HTML}{ffa3c8}%
\definecolor{colorn}{HTML}{4b0082}
\definecolor{colorn1tau}{HTML}{8b0088}

\tikzstyle{every node}=[circle,fill=black!25,minimum size=8pt,inner sep=0pt,draw=black!80]
\tikzstyle{S1}=[fill=colornm]
\tikzstyle{S2m}=[fill=colorn1tau]
\tikzstyle{S2}=[fill=colorn]
\tikzstyle{S3}=[fill=colorsqrt]
\tikzstyle{n1}=[fill=color1vertex]

\newcommand{\beginsupplementnote}{%
	\setcounter{table}{0}
	\renewcommand{\tablename}{Supplementary Table}%
	\setcounter{figure}{0}
	\renewcommand{\figurename}{Supplementary Figure}
	\setcounter{section}{0}
	\renewcommand{\thesection}{Supplementary Note \arabic{section}}
	\renewcommand{\thesubsection}{\arabic{section}.\arabic{subsection}}
}

\allowdisplaybreaks

\newtheorem{theorem}{Theorem}

\newtheorem{lemma}[theorem]{Lemma}

\graphicspath{{Figures/}}

\begin{document}
\title{Variational principle for scale-free network motifs}
\author[1,*]{Clara Stegehuis}
\author[1]{Remco van der Hofstad}
\author[1]{Johan S.H. van Leeuwaarden}
\affil[1]{Department of Mathematics and Computer Science \\Eindhoven University of Technology}

\maketitle

\begin{abstract}
	For scale-free networks with degrees following a power law with an exponent $\tau\in(2,3)$, the structures of motifs (small subgraphs) are not yet well understood. We introduce a method designed to identify the dominant structure of any given motif as the solution of an optimization problem. The unique optimizer describes the degrees of the vertices that together span the most likely motif, resulting in explicit asymptotic formulas for the motif count and its fluctuations. We then classify all motifs into two categories: motifs with small and large fluctuations. 
\end{abstract}
\section*{Introduction}\label{sec:intr}
Many real-world networks, like communication networks, social networks and biological networks, were found to be {\it scale free} with power-law degree distributions and infinite variance in the large-network limit~\cite{albert1999,faloutsos1999,jeong2000,vazquez2002,clauset2009}. The heavy tail of the power law comes with \emph{hubs}, vertices of extremely high degree. Hubs create ultra-small distances, ultra-fast information spreading and resilience against random attacks, while the average node degree is small. Attested by real-world data and explained by mathematical models, the consequences of power-law connectivity patterns now belong to the foundations of network science. 

Scale-free networks can be studied using random network models that connect vertices through edges forming power-law degree distributions. Connectivity patterns beyond edges are commonly described in terms of small subgraphs called motifs (or graphlets). There is increasing interest in algorithmic methods to count motifs~\cite{kashtan2004,grochow2007,omidi2009,schreiber2005,wernicke2006}, 
or the relation between motifs and network functions, like the spread of epidemics~\cite{house2009,noel2013,zhang2014,ritchie2015,ritchie2016,stegehuis2016}. Motifs can describe the tendency for clustering and other forms of network organization~\cite{ravasz2003,benson2016,tsourakakis2017}. 

Much existing work focuses on the classification of complex networks in terms of motif counts or frequencies. The occurrence of specific motifs such as wedges, triangles and cliques have been proven important for understanding real-world networks. 
Indeed, motif counts might vary considerably across different networks~\cite{milo2002,wuchty2003,milo2004} and any given network has a set of statistically significant motifs. Statistical relevance can be expressed by comparing real-world networks to mathematically tractable models. A popular statistic takes the motif count, subtracts the expected number of motifs in the mathematical model, and divides by the standard deviation in the mathematical model~\cite{milo2004,onnela2005,gao2017}. Such a standardized test statistic predicts whether a motif is overrepresented in comparison to some mathematical model. This comparison, however, filters out the effect of the degree sequence, the network size and possibly other features that are needed to understand the structure and frequency of motifs.

With the goal to explain the occurrence of motifs beyond counting, we develop a method to identify, for any given motif, the composition that dominates the motif count as the solution of an optimization problem. The unique optimizer describes the degrees of the vertices that together span the most likely motif, as well as predicts the leading asymptotic order for the motif count in the large-network limit. Our method can potentially be applied to various random network models, but is developed first for the hidden-variable model~\cite{goh2001,boguna2003,park2004,norros2006,chung2002,soederberg2002,caldarelli2002}, a random network model that generates graphs with power-law degrees.
Given $n$ vertices, the hidden-variable model associates to each node a hidden variable $h$ drawn from the probability density
\begin{equation}\label{eq:rhoh}
	\rho(h)=Ch^{-\tau}
\end{equation}
for some constant $C$ and $h\geq h_{\min}$. Next, conditionally on all the hidden variables, each pair of vertices is joined independently with probability
\begin{equation}\label{eq:conprob}
p(h,h')=\min(hh'/(\mu n),1).
\end{equation}
with $h$ and $h'$ the hidden variables associated with the two vertices, and $\mu$ the mean of the hidden variables. For any given motif,  we now seek for its most likely structure. The probability $P(H)$ of creating motif $H$ on $k$ uniformly chosen vertices can be written as
\begin{equation}\label{eq:pHpresent}
P(H)=\int_{\boldsymbol{h}}\Prob{H\text{ on }h_1,\ldots,h_k}\Prob{h_1,\ldots,h_k}\dd \boldsymbol{h},
\end{equation}
where the integral is over all possible hidden-variable sequences on $k$ vertices, with 
$\boldsymbol{h}=(h_1,\ldots,h_k)$ and $\Prob{h_1,\ldots,h_k}$ the density that a randomly chosen set of $k$ hidden variables is proportional to $h_1,\ldots,h_k$. The degree of a node is asymptotically Poisson distributed with its hidden variable as mean~\cite{stegehuis2017}, so~\eqref{eq:pHpresent} can be interpreted as a sum over all possible degree sequences. Therefore, our optimization method then needs to settle the following trade-off, inherently present in power-law networks: On the one hand, large-degree vertices contribute substantially to the number of motifs, because they are highly connected, and therefore participate in many motifs. On the other hand, large-degree vertices are by definition rare. This should be contrasted with lower-degree vertices that occur more frequently, but take part in fewer connections and hence fewer motifs. Therefore, our method give rise to a certain `variational principle', because it finds the selection of vertices with specific degrees that together `optimize' this trade-off and hence maximize the expected number of such motifs. 

We leverage the optimization method in two ways. First, we derive sharp expressions for the motif counts in the large-network limit in terms of the network size and the power-law exponent. Second, we use the method to identify the fluctuations of motif counts. 

We present two versions of the method that we call free and typical variation. Free variation corresponds to computing the average number of motifs over many samples of the random network model. Typical variation corresponds to the number of motifs in one single instance of the random graph model. Remarkably, for $\tau\in(2,3)$ these can be rather different. After that, we apply the method to study motif count fluctuations. Finally, we provide a case study where we investigate the presence of motifs in some real-world network data.

\section*{Results}
\subsection*{Free variation}\label{sec:freevar}
We first show that only hidden-variable sequences $\boldsymbol{h}$ with hidden variables of specific orders give the largest contribution to~\eqref{eq:pHpresent}.
Write the hidden variables as $h_i\propto{n^{\alpha_i}}$ for some $\alpha_i\geq 0$ for all $i$.
Then, using~\eqref{eq:conprob}, the probability that motif $H$ exists on vertices with hidden variables $\boldsymbol{h} = (n^{\alpha_1},\ldots,n^{\alpha_k})$ satisfies
\begin{equation}\label{eq:phsubx}
\Prob{H\text{ on }{\boldsymbol{h}}}\propto\prod_{(v_i,v_j)\in E_H:\alpha_{i}+\alpha_{j}<1}n^{\alpha_{i}+\alpha_{j}-1}.
\end{equation}

The hidden variables are an i.i.d.\ sample from a power-law distribution, so that the probability that $k$ uniformly chosen hidden variables satisfy $(h_1,\ldots,h_k)\propto (n^{\alpha_1},\ldots,n^{\alpha_k})$ is of the order $n^{(1-\tau)\sum_i\alpha_i}$ (see Supplementary Material 2). Taking the product of this with~\eqref{eq:phsubx} shows that the maximum contribution to the summand in~\eqref{eq:pHpresent} is obtained for those $\alpha_i\geq 0$ that maximize the exponent 
\begin{equation}\label{eq:maxalphx}
(1-\tau)\sum_{i}\alpha_i +\sum_{(i,j)\in E_H: \alpha_i+\alpha_j<1}\left(\alpha_i+\alpha_j-1\right),
\end{equation} 
which is a piecewise-linear function in $\alpha$. In Supplementary Material 2, we show that the maximizer of this optimization problem satisfies $\alpha_i\in\{0,\tfrac{1}{2},1\}$ for all $i$. 
Thus, the maximal value of~\eqref{eq:maxalphx} is attained by partitioning the vertices of $H$ into the sets $S_1,S_2,S_3$ such that vertices in $S_1$ have $\alpha_i=0$, vertices in $S_2$ have $\alpha_i=1$ and vertices in $S_3$ have $\alpha_i=\tfrac{1}{2}$. Then, the edges with $\alpha_i+\alpha_j<1$ are edges inside $S_1$ and edges between $S_1$ and $S_3$. If we denote the number of edges inside $S_1$ by $E_{S_1}$ and the number of edges between $S_1$ and $S_3$ by $E_{S_1,S_3}$, then maximizing~\eqref{eq:maxalphx} is equivalent to maximizing
	\begin{equation}\label{eq:maxeqx}
	B_f(H)=\max_{\mathcal{P}}\left[\abs{S_1}-\abs{S_2}-\frac{2E_{S_1}+E_{S_1,S_3}}{\tau-1}\right]
	\end{equation}
	over all partitions $\mathcal{P}$ of the vertices of $H$ into $S_1,S_2,S_3$. This gives the following theorem (a more elaborate version is proven in Supplementary Material 2):

\begin{theorem}[Expected motif count]\label{thm:expmotifsx}
	Let $H$ be a motif on $k$ vertices such that the solution to~\eqref{eq:maxeqx} is unique. 
	As $n\to\infty$, the expected number of motifs $H$ grows as
	\begin{equation}
\Exp{N(H)} = n^kP(H) \propto n^{\frac{3-\tau}{2}k+\frac{\tau-1}{2}B_f(H)},
	\end{equation}
	and is thus fully determined by the partition $\mathcal{P}^*$ that optimizes \eqref{eq:maxeqx}. 
\end{theorem}

Theorem~\ref{thm:expmotifsx} implies that the expected number of motifs is dominated by motifs on vertices with hidden variables (and thus degrees) of specific orders of magnitude: constant degrees, degrees proportional to $\sqrt{n}$ or degrees proportional to $n$. 
Figure~\ref{fig:graphlet4exp} and Figure~\ref{fig:motif5exp} show the partitions $\mathcal{P}^*$ that dominate the expected number of motifs on three, four and five vertices. 

\begin{figure}[tb]
	\begin{subfigure}{0.17\linewidth}
		\centering
		\begin{tikzpicture}
		\tikzstyle{edge} = [draw,thick,-]
		\node[S3] (a) at (0,0) {};
		\node[S3] (b) at (1,0) {};
		\node[S3] (c) at (1,1) {};
		\draw[edge] (a)--(b);
		\draw[edge] (c)--(b);
		\draw[edge] (a)--(c);
		\end{tikzpicture}	
		\caption{$n^{\frac{3}{2}(3-2\tau)}$}
		\label{var:fig:K3exp}
	\end{subfigure}
	\begin{subfigure}{0.17\linewidth}
		\centering
		\begin{tikzpicture}
		\tikzstyle{edge} = [draw,thick,-]
		\node[n1] (a) at (0,0) {};
		\node[n1] (b) at (1,1) {};
		\node[S2] (c) at (0,1) {};
		\draw[edge] (a)--(c);
		\draw[edge] (c)--(b);
		\end{tikzpicture}	
		\caption{$n^{4-\tau}$}
		\label{var:fig:3wedge}
	\end{subfigure}
	\begin{subfigure}{0.17\linewidth}
		\centering
		\begin{tikzpicture}
		\tikzstyle{edge} = [draw,thick,-]
		\node[S3] (a) at (0,0) {};
		\node[S3] (b) at (1,0) {};
		\node[S3] (c) at (0,1) {};
		\node[S3] (d) at (1,1) {};
		\draw[edge] (a)--(b);
		\draw[edge] (c)--(b);
		\draw[edge] (d)--(b);
		\draw[edge] (a)--(c);
		\draw[edge] (a)--(d);
		\draw[edge] (c)--(d);
		\end{tikzpicture}	
		\caption{$n^{6-2\tau}$}
		\label{var:fig:K4exp}
	\end{subfigure}
	\begin{subfigure}{0.17\linewidth}
		\centering
		\begin{tikzpicture}
		\tikzstyle{edge} = [draw,thick,-]
		\node (a) at (0,0) {};
		\node (b) at (1,0) {};
		\node (c) at (0,1) {};
		\node (d) at (1,1) {};
		\draw[edge] (a)--(b);
		\draw[edge] (c)--(b);
		\draw[edge] (d)--(b);
		\draw[edge] (a)--(c);
		\draw[edge] (c)--(d);
		\end{tikzpicture}	
		\caption{$n^{6-2\tau}\log(n)$}
		\label{var:fig:squareextraexp}
	\end{subfigure}
	
	\begin{subfigure}{0.17\linewidth}
		\centering
		\begin{tikzpicture}
		\tikzstyle{edge} = [draw,thick,-]
		\node[] (a) at (0,0) {};
		\node[] (b) at (1,0) {};
		\node[] (c) at (0,1) {};
		\node[] (d) at (1,1) {};
		\draw[edge] (a)--(b);
		\draw[edge] (d)--(b);
		\draw[edge] (a)--(c);
		\draw[edge] (c)--(d);
		\end{tikzpicture}	
		\caption{$n^{6-2\tau}\log(n)$}
		\label{var:fig:squareexp}
	\end{subfigure}
	\begin{subfigure}{0.17\linewidth}
		\centering
		\begin{tikzpicture}
		\tikzstyle{edge} = [draw,thick,-]
		\node[n1] (a) at (0,0) {};
		\node[n1] (b) at (1,0) {};
		\node[S2] (c) at (0,1) {};
		\node[n1] (d) at (1,1) {};
		\draw[edge] (c)--(b);
		\draw[edge] (d)--(b);
		\draw[edge] (a)--(c);
		\draw[edge] (c)--(d);
		\end{tikzpicture}	
		\caption{$n^{4-\tau}$}
		\label{var:fig:pawexp}
	\end{subfigure}
	\begin{subfigure}{0.17\linewidth}
		\centering
		\begin{tikzpicture}
		\tikzstyle{edge} = [draw,thick,-]
		\node[n1] (a) at (0,0) {};
		\node[n1] (b) at (1,0) {};
		\node[S2] (c) at (0,1) {};
		\node[n1] (d) at (1,1) {};
		\draw[edge] (c)--(b);
		\draw[edge] (a)--(c);
		\draw[edge] (c)--(d);
		\end{tikzpicture}	
		\caption{$n^{4-\tau}$}
		\label{var:fig:wedge4exp}
	\end{subfigure}
	\begin{subfigure}{0.17\linewidth}
		\centering
		\begin{tikzpicture}
		\tikzstyle{edge} = [draw,thick,-]
		\node[n1] (a) at (0,0) {};
		\node[n1] (b) at (1,0) {};
		\node (c) at (0,1) {};
		\node (d) at (1,1) {};
		\draw[edge] (d)--(b);
		\draw[edge] (a)--(c);
		\draw[edge] (c)--(d);
		\end{tikzpicture}	
		\caption{$n^{4-\tau}\log(n)$}
		\label{var:fig:4pathexp}
	\end{subfigure}
	\hspace{0.5cm}
	\begin{subfigure}{0.2\linewidth}
		\begin{tikzpicture}
		\node[S2,label={[label distance=0.05cm]0:$n$}] (a) at (0,1.5) {};
		\node[S3,label={[label distance=0cm]0:$\sqrt{n}$}] (b) at (0,1) {};
		\node[n1,label={[label distance=0.05cm]0:$1$}] (c) at (0,0.5) {};
		\node[label={[label distance=0.05cm]0:non-unique}] (d) at (0,0) {};
		\end{tikzpicture}
	\end{subfigure}
	\caption{Scaling of the expected number of motifs on 3 and 4 vertices in $n$, where the vertex color indicates the dominant vertex degree. Vertices where the optimizer is not unique are gray.}
	\label{fig:graphlet4exp}
\end{figure}
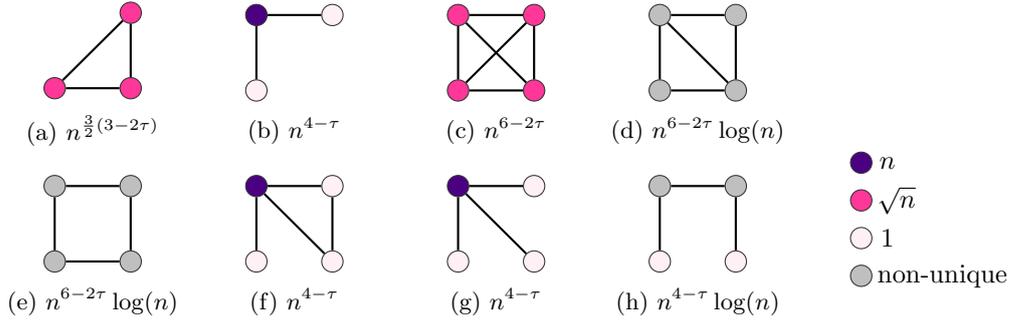
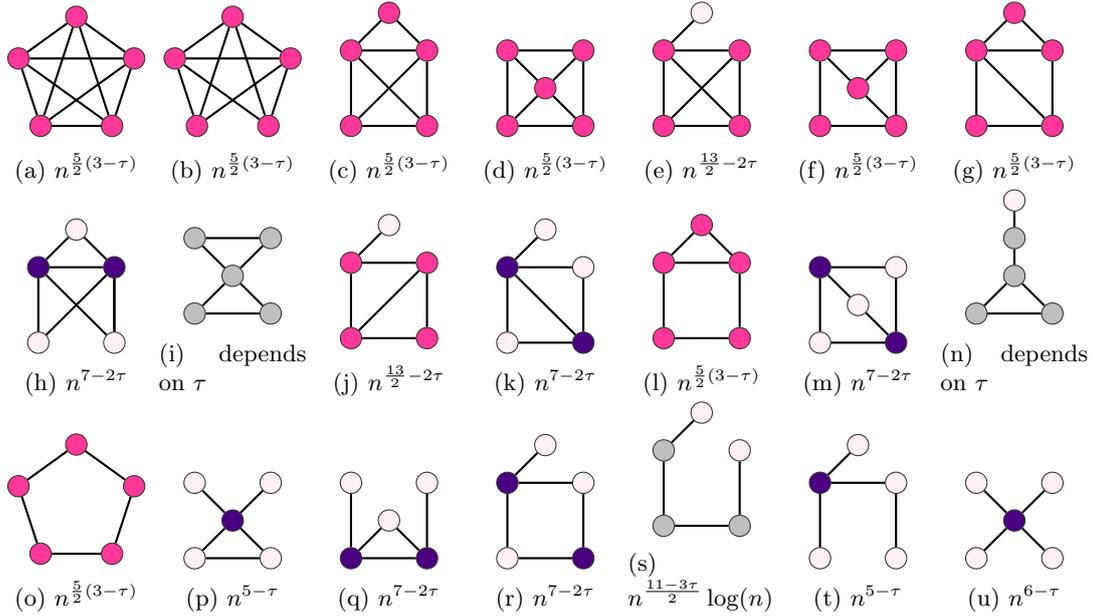
\begin{figure}[tb]
	\centering
	\begin{subfigure}[b]{0.13\linewidth}
		\centering
		\begin{tikzpicture}
		\tikzstyle{edge} = [draw,thick,-]
		\node[S3] (a) at (90:0.8) {};
		\node[S3] (b) at (162:0.8) {};
		\node[S3] (c) at (234:0.8) {};
		\node[S3] (d) at (306:0.8) {};
		\node[S3] (e) at (378:0.8) {};
		\draw[edge] (a)--(b);
		\draw[edge] (c)--(b);
		\draw[edge] (d)--(b);
		\draw[edge] (a)--(c);
		\draw[edge] (a)--(d);
		\draw[edge] (c)--(e);
		\draw[edge] (e)--(d);
		\draw[edge] (b)--(e);
		\draw[edge] (a)--(e);
		\draw[edge] (c)--(d);
		\end{tikzpicture}	
		\caption{$n^{\frac{5}{2}(3-\tau)}$}
	\end{subfigure}
	\begin{subfigure}[b]{0.13\linewidth}
		\centering
		\begin{tikzpicture}
		\tikzstyle{edge} = [draw,thick,-]
		\node[S3] (a) at (90:0.8) {};
		\node[S3] (b) at (162:0.8) {};
		\node[S3] (c) at (234:0.8) {};
		\node[S3] (d) at (306:0.8) {};
		\node[S3] (e) at (378:0.8) {};
		\draw[edge] (a)--(b);
		\draw[edge] (c)--(b);
		\draw[edge] (d)--(b);
		\draw[edge] (a)--(c);
		\draw[edge] (a)--(d);
		\draw[edge] (c)--(e);
		\draw[edge] (e)--(d);
		\draw[edge] (b)--(e);
		\draw[edge] (a)--(e);
		\end{tikzpicture}	
		\caption{$n^{\frac{5}{2}(3-\tau)}$}
	\end{subfigure}
	\begin{subfigure}[b]{0.13\linewidth}
		\centering
		\begin{tikzpicture}
		\tikzstyle{edge} = [draw,thick,-]
		\node[S3] (a) at (0,0) {};
		\node[S3] (b) at (1,0) {};
		\node[S3] (c) at (0,1) {};
		\node[S3] (d) at (1,1) {};
		\node[S3] (e) at (0.5,1.5) {};
		\draw[edge] (a)--(b);
		\draw[edge] (c)--(b);
		\draw[edge] (d)--(b);
		\draw[edge] (a)--(c);
		\draw[edge] (a)--(d);
		\draw[edge] (c)--(e);
		\draw[edge] (d)--(e);
		\draw[edge] (c)--(d);
		\end{tikzpicture}	
		\caption{$n^{\frac{5}{2}(3-\tau)}$}
	\end{subfigure}
	\begin{subfigure}[b]{0.13\linewidth}
		\centering
		\begin{tikzpicture}
		\tikzstyle{edge} = [draw,thick,-]
		\node[S3] (a) at (0,0) {};
		\node[S3] (b) at (1,0) {};
		\node[S3] (c) at (0,1) {};
		\node[S3] (d) at (1,1) {};
		\node[S3] (e) at (0.5,0.5) {};
		\draw[edge] (a)--(b);
		\draw[edge] (c)--(a);
		\draw[edge] (b)--(e);
		\draw[edge] (a)--(e);
		\draw[edge] (b)--(d);
		\draw[edge] (c)--(e);
		\draw[edge] (d)--(e);
		\draw[edge] (c)--(d);
		\end{tikzpicture}	
		\caption{$n^{\frac{5}{2}(3-\tau)}$}
	\end{subfigure}
	\begin{subfigure}[b]{0.13\linewidth}
		\centering
		\begin{tikzpicture}
		\tikzstyle{edge} = [draw,thick,-]
		\node[S3] (a) at (0,0) {};
		\node[S3] (b) at (1,0) {};
		\node[S3] (c) at (0,1) {};
		\node[S3] (d) at (1,1) {};
		\node[n1] (e) at (0.5,1.5) {};
		\draw[edge] (a)--(b);
		\draw[edge] (c)--(b);
		\draw[edge] (d)--(b);
		\draw[edge] (a)--(c);
		\draw[edge] (a)--(d);
		\draw[edge] (c)--(e);
		\draw[edge] (c)--(d);
		\end{tikzpicture}	
		\caption{$n^{\frac{13}{2}-2\tau}$}
	\end{subfigure}
	\begin{subfigure}[b]{0.13\linewidth}
		\centering
		\begin{tikzpicture}
		\tikzstyle{edge} = [draw,thick,-]
		\node[S3] (a) at (0,0) {};
		\node[S3] (b) at (1,0) {};
		\node[S3] (c) at (0,1) {};
		\node[S3] (d) at (1,1) {};
		\node[S3] (e) at (0.5,0.5) {};
		\draw[edge] (a)--(b);
		\draw[edge] (c)--(a);
		\draw[edge] (b)--(e);
		\draw[edge] (b)--(d);
		\draw[edge] (c)--(e);
		\draw[edge] (c)--(d);
		\draw[edge] (e)--(d);
		\end{tikzpicture}	
		\caption{$n^{\frac{5}{2}(3-\tau)}$}
	\end{subfigure}		
	\begin{subfigure}[b]{0.13\linewidth}
		\centering
		\begin{tikzpicture}
		\tikzstyle{edge} = [draw,thick,-]
		\node[S3] (a) at (0,0) {};
		\node[S3] (b) at (1,0) {};
		\node[S3] (c) at (0,1) {};
		\node[S3] (d) at (1,1) {};
		\node[S3] (e) at (0.5,1.5) {};
		\draw[edge] (a)--(b);
		\draw[edge] (c)--(b);
		\draw[edge] (d)--(b);
		\draw[edge] (a)--(c);
		\draw[edge] (c)--(e);
		\draw[edge] (d)--(e);
		\draw[edge] (c)--(d);
		\end{tikzpicture}	
		\caption{$n^{\frac{5}{2}(3-\tau)}$}
	\end{subfigure}

	\begin{subfigure}[b]{0.13\linewidth}
		\centering
		\begin{tikzpicture}
		\tikzstyle{edge} = [draw,thick,-]
		\node[n1] (a) at (0,0) {};
		\node[n1] (b) at (1,0) {};
		\node[S2] (c) at (0,1) {};
		\node[S2] (d) at (1,1) {};
		\node[n1] (e) at (0.5,1.5) {};
		\draw[edge] (d)--(b);
		\draw[edge] (c)--(b);
		\draw[edge] (d)--(b);
		\draw[edge] (a)--(c);
		\draw[edge] (a)--(d);
		\draw[edge] (c)--(e);
		\draw[edge] (d)--(e);
		\draw[edge] (c)--(d);
		\end{tikzpicture}	
		\caption{$n^{7-2\tau}$}
		\label{var:fig:K23}
	\end{subfigure}
	\begin{subfigure}[b]{0.13\linewidth}
		\centering
		\begin{tikzpicture}
		\tikzstyle{edge} = [draw,thick,-]
		\node (a) at (0,0) {};
		\node (b) at (1,0) {};
		\node (c) at (0,1) {};
		\node (d) at (1,1) {};
		\node (e) at (0.5,0.5) {};
		\draw[edge] (a)--(b);
		\draw[edge] (e)--(a);
		\draw[edge] (b)--(e);
		\draw[edge] (e)--(d);
		\draw[edge] (c)--(e);
		\draw[edge] (c)--(d);
		\end{tikzpicture}	
		\caption{depends on $\tau$}
		\label{var:fig:bowtie}
	\end{subfigure}
	\begin{subfigure}[b]{0.13\linewidth}
		\centering
		\begin{tikzpicture}
		\tikzstyle{edge} = [draw,thick,-]
		\node[S3] (a) at (0,0) {};
		\node[S3] (b) at (1,0) {};
		\node[S3] (c) at (0,1) {};
		\node[S3] (d) at (1,1) {};
		\node[n1] (e) at (0.5,1.5) {};
		\draw[edge] (a)--(b);
		\draw[edge] (d)--(b);
		\draw[edge] (a)--(c);
		\draw[edge] (a)--(d);
		\draw[edge] (c)--(e);
		\draw[edge] (c)--(d);
		\end{tikzpicture}	
		\caption{$n^{\frac{13}{2}-2\tau}$}
	\end{subfigure}
	\begin{subfigure}[b]{0.13\linewidth}
		\centering
		\begin{tikzpicture}
		\tikzstyle{edge} = [draw,thick,-]
		\node[n1] (a) at (0,0) {};
		\node[S2] (b) at (1,0) {};
		\node[S2] (c) at (0,1) {};
		\node[n1] (d) at (1,1) {};
		\node[n1] (e) at (0.5,1.5) {};
		\draw[edge] (a)--(b);
		\draw[edge] (d)--(b);
		\draw[edge] (a)--(c);
		\draw[edge] (b)--(c);
		\draw[edge] (c)--(e);
		\draw[edge] (c)--(d);
		\end{tikzpicture}	
		\caption{$n^{7-2\tau}$}
	\end{subfigure}
	\begin{subfigure}[b]{0.13\linewidth}
		\centering
		\begin{tikzpicture}
		\tikzstyle{edge} = [draw,thick,-]
		\node[S3] (a) at (0,0) {};
		\node[S3] (b) at (1,0) {};
		\node[S3] (c) at (0,1) {};
		\node[S3] (d) at (1,1) {};
		\node[S3] (e) at (0.5,1.5) {};
		\draw[edge] (a)--(b);
		\draw[edge] (d)--(b);
		\draw[edge] (a)--(c);
		\draw[edge] (e)--(d);
		\draw[edge] (c)--(e);
		\draw[edge] (c)--(d);
		\end{tikzpicture}	
		\caption{$n^{\frac{5}{2}(3-\tau)}$}
	\end{subfigure}	
	\begin{subfigure}[b]{0.13\linewidth}
		\centering
		\begin{tikzpicture}
		\tikzstyle{edge} = [draw,thick,-]
		\node[n1] (a) at (0,0) {};
		\node[S2] (b) at (1,0) {};
		\node[S2] (c) at (0,1) {};
		\node[n1] (d) at (1,1) {};
		\node[n1] (e) at (0.5,0.5) {};
		\draw[edge] (a)--(b);
		\draw[edge] (c)--(a);
		\draw[edge] (b)--(e);
		\draw[edge] (b)--(d);
		\draw[edge] (c)--(e);
		\draw[edge] (c)--(d);
		\end{tikzpicture}	
		\caption{$n^{7-2\tau}$}
		\label{var:fig:m5dom}
	\end{subfigure}
	\begin{subfigure}[b]{0.13\linewidth}
		\centering
		\begin{tikzpicture}
		\tikzstyle{edge} = [draw,thick,-]
		\node (a) at (0,0) {};
		\node (b) at (1,0) {};
		\node (c) at (0.5,1) {};
		\node[n1] (d) at (0.5,1.5) {};
		\node (e) at (0.5,0.5) {};
		\draw[edge] (a)--(b);
		\draw[edge] (e)--(a);
		\draw[edge] (b)--(e);
		\draw[edge] (e)--(c);
		\draw[edge] (c)--(d);
		\end{tikzpicture}	
		\caption{depends on $\tau$}
	\end{subfigure}	

	\begin{subfigure}[b]{0.13\linewidth}
		\centering
		\begin{tikzpicture}
		\tikzstyle{edge} = [draw,thick,-]
		\node[S3] (a) at (90:0.8) {};
		\node[S3] (b) at (162:0.8) {};
		\node[S3] (c) at (234:0.8) {};
		\node[S3] (d) at (306:0.8) {};
		\node[S3] (e) at (378:0.8) {};
		\draw[edge] (a)--(b);
		\draw[edge] (c)--(b);
		\draw[edge] (d)--(c);
		\draw[edge] (a)--(e);
		\draw[edge] (d)--(e);
		\end{tikzpicture}	
		\caption{$n^{\frac{5}{2}(3-\tau)}$}
	\end{subfigure}
	\begin{subfigure}[b]{0.13\linewidth}
		\centering
		\begin{tikzpicture}
		\tikzstyle{edge} = [draw,thick,-]
		\node[n1] (a) at (0,0) {};
		\node[n1] (b) at (1,0) {};
		\node[n1] (c) at (0,1) {};
		\node[n1] (d) at (1,1) {};
		\node[S2] (e) at (0.5,0.5) {};
		\draw[edge] (a)--(b);
		\draw[edge] (e)--(a);
		\draw[edge] (b)--(e);
		\draw[edge] (e)--(d);
		\draw[edge] (c)--(e);
		\end{tikzpicture}	
		\caption{$n^{5-\tau}$}
	\end{subfigure}
	\begin{subfigure}[b]{0.13\linewidth}
		\centering
		\begin{tikzpicture}
		\tikzstyle{edge} = [draw,thick,-]
		\node[S2] (a) at (0,0) {};
		\node[S2] (b) at (1,0) {};
		\node[n1] (c) at (0,1) {};
		\node[n1] (d) at (1,1) {};
		\node[n1] (e) at (0.5,0.5) {};
		\draw[edge] (a)--(b);
		\draw[edge] (e)--(a);
		\draw[edge] (b)--(e);
		\draw[edge] (b)--(d);
		\draw[edge] (c)--(a);
		\end{tikzpicture}	
		\caption{$n^{7-2\tau}$}
	\end{subfigure}
	\begin{subfigure}[b]{0.13\linewidth}
		\centering
		\begin{tikzpicture}
		\tikzstyle{edge} = [draw,thick,-]
		\node[n1] (a) at (0,0) {};
		\node[S2] (b) at (1,0) {};
		\node[S2] (c) at (0,1) {};
		\node[n1] (d) at (1,1) {};
		\node[n1] (e) at (0.5,1.5) {};
		\draw[edge] (a)--(b);
		\draw[edge] (d)--(b);
		\draw[edge] (a)--(c);
		\draw[edge] (c)--(e);
		\draw[edge] (c)--(d);
		\end{tikzpicture}	
		\caption{$n^{7-2\tau}$}
	\end{subfigure}	
	\begin{subfigure}[b]{0.13\linewidth}
		\centering
		\begin{tikzpicture}
		\tikzstyle{edge} = [draw,thick,-]
		\node (a) at (0,0) {};
		\node (b) at (1,0) {};
		\node (c) at (0,1) {};
		\node[n1] (d) at (1,1) {};
		\node[n1] (e) at (0.5,1.5) {};
		\draw[edge] (a)--(b);
		\draw[edge] (d)--(b);
		\draw[edge] (a)--(c);
		\draw[edge] (c)--(e);
		\end{tikzpicture}	
		\caption{$n^{\frac{11-3\tau}{2}}\log(n)$}
		\label{var:fig:5path}
	\end{subfigure}
	\begin{subfigure}[b]{0.13\linewidth}
		\centering
		\begin{tikzpicture}
		\tikzstyle{edge} = [draw,thick,-]
		\node[n1] (a) at (0,0) {};
		\node[n1] (b) at (1,0) {};
		\node[S2] (c) at (0,1) {};
		\node[n1] (d) at (1,1) {};
		\node[n1] (e) at (0.5,1.5) {};
		\draw[edge] (d)--(b);
		\draw[edge] (a)--(c);
		\draw[edge] (c)--(e);
		\draw[edge] (c)--(d);
		\end{tikzpicture}	
		\caption{$n^{5-\tau}$}
		\label{var:fig:clawdif}
	\end{subfigure}
	\begin{subfigure}[b]{0.13\linewidth}
		\centering
		\begin{tikzpicture}
		\tikzstyle{edge} = [draw,thick,-]
		\node[n1] (a) at (0,0) {};
		\node[n1] (b) at (1,0) {};
		\node[n1] (c) at (0,1) {};
		\node[n1] (d) at (1,1) {};
		\node[S2] (e) at (0.5,0.5) {};
		\draw[edge] (e)--(a);
		\draw[edge] (b)--(e);
		\draw[edge] (e)--(d);
		\draw[edge] (c)--(e);
		\end{tikzpicture}	
		\caption{$n^{6-\tau}$}
		\label{var:fig:4claw}
	\end{subfigure}
	\caption{Scaling of the expected number of motifs on 5 vertices in $n$, where the vertex color indicates the dominant vertex degree, as in Fig.~\ref{fig:graphlet4exp}.}
	\label{fig:motif5exp}
\end{figure}

\subsection*{Typical variation}\label{sec:typ}
The largest degrees (hubs) in typical samples of the hidden-variable model scale as $n^{1/(\tau-1)}$ with high probability. The expected number of motifs, however, may be dominated by network samples where the largest degree is proportional to $n$ (see Theorem~\ref{thm:expmotifsx}). These samples contain many motifs because of the high degrees, and therefore contribute significantly to the expectation. Nevertheless, the probability of observing such a network tends to zero as $n$ grows large. We therefore now adapt the variational principle with the goal to characterize the typical motif structure and hence the typical number of motifs.

We again assume degrees to be proportional to $n^{\alpha_i}$, but now limit to degree sequences where the maximal degree is of order $n^{1/(\tau-1)}$, the natural cutoff in view of the typical hub degrees. 
The dominant typical motif structure is then obtained by maximizing~\eqref{eq:maxalphx}, with the additional constraint that $\alpha_{i}\leq\tfrac{1}{\tau-1}$. In Supplementary Material 3 we show that the possible optimizers are one of four values $\alpha_i\in\{0,\tfrac{\tau-2}{\tau-1},\tfrac{1}{2},\tfrac{1}{\tau-1}\}$, and obtain an optimization problem similar to~\eqref{eq:maxeqx}.

This shows that the typical degree of a motif is of constant order or proportional to $n^{1/(\tau-1)}$, $\sqrt{n}$ or $n^{(\tau-2)/(\tau-1)}$. Figure~\ref{fig:graphlet4} and Supplementary Figure~1 show the most likely motifs on three, four and five vertices.
Observe that the dominant structures and the number of motifs of Fig.~\ref{fig:graphlet4exp} and~\ref{fig:graphlet4} may differ. 
For example, the scaling of the expected number of claws (Fig.~\ref{fig:graphlet4exp}b) and the typical number of claws (Fig.~\ref{fig:graphlet4}b) is different. This is caused by the left upper vertex that has degree proportional to $n$ in the free dominant structure, whereas its typical degree is proportional to $n^{1/(\tau-1)}$.  
Only when the solution to~\eqref{eq:maxeqx} does not involve hub vertices, the two scalings coincide. Hub vertices in the dominant structure give a major contribution to the motif count. While typical hub degrees scale as $n^{1/(\tau-1)}$, expected hub degrees may be much larger, causing the number of such motifs with hubs to scale faster in the free variation setting than in the typical variation setting. This indicates that the average and median motif count can differ dramatically. 

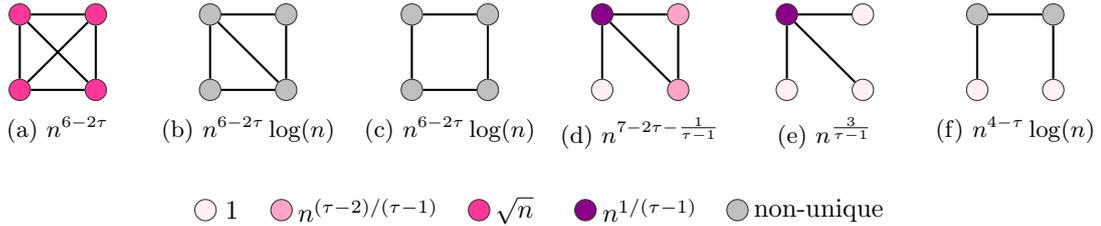
\begin{figure}[tb]
\centering
	\begin{subfigure}[t]{0.15\linewidth}
	\centering
	\begin{tikzpicture}
	\tikzstyle{edge} = [draw,thick,-]
	\node[S3] (a) at (0,0) {};
	\node[S3] (b) at (1,0) {};
	\node[S3] (c) at (0,1) {};
	\node[S3] (d) at (1,1) {};
	\draw[edge] (a)--(b);
	\draw[edge] (c)--(b);
	\draw[edge] (d)--(b);
	\draw[edge] (a)--(c);
	\draw[edge] (a)--(d);
	\draw[edge] (c)--(d);
	\end{tikzpicture}	
	\caption{$n^{6-2\tau}$}
	\label{mot:fig:K4sub}
\end{subfigure}
\begin{subfigure}[t]{0.17\linewidth}
	\centering
	\begin{tikzpicture}
	\tikzstyle{edge} = [draw,thick,-]
	\node (a) at (0,0) {};
	\node (b) at (1,0) {};
	\node (c) at (0,1) {};
	\node (d) at (1,1) {};
	\draw[edge] (a)--(b);
	\draw[edge] (c)--(b);
	\draw[edge] (d)--(b);
	\draw[edge] (a)--(c);
	\draw[edge] (c)--(d);
	\end{tikzpicture}	
	\caption{$n^{6-2\tau}\log(n)$}
	\label{mot:fig:squareextrasub}
\end{subfigure}
\begin{subfigure}[t]{0.17\linewidth}
	\centering
	\begin{tikzpicture}
	\tikzstyle{edge} = [draw,thick,-]
	\node (a) at (0,0) {};
	\node (b) at (1,0) {};
	\node (c) at (0,1) {};
	\node (d) at (1,1) {};
	\draw[edge] (a)--(b);
	\draw[edge] (d)--(b);
	\draw[edge] (a)--(c);
	\draw[edge] (c)--(d);
	\end{tikzpicture}	
	\caption{$n^{6-2\tau}\log(n)$}
	\label{mot:fig:squaresub}
\end{subfigure}
\begin{subfigure}[t]{0.15\linewidth}
	\centering
	\begin{tikzpicture}
	\tikzstyle{edge} = [draw,thick,-]
	\node[n1] (a) at (0,0) {};
	\node[S1] (b) at (1,0) {};
	\node[S2m] (c) at (0,1) {};
	\node[S1] (d) at (1,1) {};
	\draw[edge] (c)--(b);
	\draw[edge] (d)--(b);
	\draw[edge] (a)--(c);
	\draw[edge] (c)--(d);
	\end{tikzpicture}	
	\caption{$n^{7-2\tau-\frac{1}{\tau-1}}$}
\end{subfigure}
\begin{subfigure}[t]{0.16\linewidth}
	\centering
	\begin{tikzpicture}
	\tikzstyle{edge} = [draw,thick,-]
	\node[n1] (a) at (0,0) {};
	\node[n1] (b) at (1,0) {};
	\node[S2m] (c) at (0,1) {};
	\node[n1] (d) at (1,1) {};
	\draw[edge] (c)--(b);
	\draw[edge] (a)--(c);
	\draw[edge] (c)--(d);
	\end{tikzpicture}	
	\caption{$n^{\frac{3}{\tau-1}}$}
\end{subfigure}
\begin{subfigure}[t]{0.16\linewidth}
	\centering
	\begin{tikzpicture}
	\tikzstyle{edge} = [draw,thick,-]
	\node[n1] (a) at (0,0) {};
	\node[n1] (b) at (1,0) {};
	\node (c) at (0,1) {};
	\node (d) at (1,1) {};
	\draw[edge] (d)--(b);
	\draw[edge] (a)--(c);
	\draw[edge] (c)--(d);
	\end{tikzpicture}	
	\caption{$n^{4-\tau}\log(n)$}
	\label{mot:fig:wedge4mot}
\end{subfigure}

\vspace{-0.2cm}
\begin{subfigure}{\linewidth}
	\centering
	\begin{tikzpicture}
	\node[S2m,label={[label distance=0.05cm]0:$n^{1/(\tau-1)}$}] (a) at (5,0) {};
	\node[S3,label={[label distance=0cm]0:$\sqrt{n}$}] (b) at (3.6,0) {};
	\node[S1,label={[label distance=0.05cm]0:$n^{(\tau-2)/(\tau-1)}$}] (c) at (1,0) {};
	\node[n1,label={[label distance=0.05cm]0:$1$}] (c) at (0,0) {};
	\node[label={[label distance=0.05cm]0:non-unique}] (d) at (7,0) {};
	\end{tikzpicture}
\end{subfigure}
\vspace{-0.8cm}
	\caption{Typical scaling of the number of motifs on three or four vertices in $n$. The vertex color indicates the dominant vertex degree.}
	\label{fig:graphlet4}
\end{figure}

\paragraph*{Graphlets}
It is also possible to only count the number of times $H$ appears as an induced subgraph, also called graphlet counting. This means that an edge that is not present in graphlet $H$, should also be absent in the network motif. 
In Supplementary Material 4 we classify the expected and typical number of graphlets with a similar variational principle as for motifs. Supplementary Figure 2 shows the typical behavior of graphlets on 4 vertices. This figure also shows that graphlet counting is more detailed than motif counting. For example, counting all square motifs is equivalent to counting all graphlets that contain the square as an induced subgraph: the square, the diamond and $K_4$. Indeed, we obtain that the number of square motifs scales as $n^{6-2\tau}\log(n)$ by adding the number of square, diamond and $K_4$ graphlets from Supplementary Figure 2. This shows that most square motifs are actually the diamond graphlets of Supplementary Figure 2b. Thus, graphlet counting gives more detailed information than motif counting.

\subsection*{Fluctuations}\label{sec:selfavg}
Self-averaging network properties have relative fluctuations that tend to zero as the network size $n$ tends to infinity. Several physical quantities in for example Ising models, fluid models and properties of the galaxy display non-self-averaging behavior~\cite{wiseman1995,wiseman1998,aharony1996,das2017,ostilli2014,pastur1991,labini2009}. We consider motif counts $N(H)$ and call $N(H)$ self-averaging when $\Var{N(H)}/\Exp{N(H)}^2\to{0}$ as $n\to\infty$. Essential understanding of $N(H)$ can then be obtained by taking a large network sample, since the sample-to-sample fluctuations vanish in the large-network limit. In contrast, if $\Var{N(H)}/\Exp{N(H)}^2$ approaches a constant or tends to infinity as $n\to\infty$, the motif count is called non-self-averaging, in which case $N(H)$ shows (too) strong sample-to-sample fluctuations that cannot be mitigated by taking more network samples. 

Our variational principle facilitates a systematic study of such fluctuations, and leads to a classification into self-averaging and non-self-averaging for all motifs $H$. It turns out that whether $N(H)$ is self-averaging or not depends on the power-law exponent $\tau$ and the dominant structure of $H$. 
We also show that non-self-averaging behavior of motif counts may not have the intuitive explanation described above. In some cases, motif counts in two instances are similar with high probability, but rare network samples behave differently, causing the motif count to be non-self-averaging. Thus, the classification of network motifs into self-averaging and non-self-averaging motifs does not give a complete picture of the motif count fluctuations. We therefore further divide the non-self-averaging motifs into two classes based on the type of fluctuations in the motif counts.

For a given motif $H$, let $H_1,\dots,H_m$ denote all possible motifs that can be constructed by merging two copies of $H$ at one or more vertices. 
We can then write the variance of the motif count as (see \cite{frank1979,picard2008,matias2006,ostilli2014}  and the Methods section)
\begin{equation}\label{eq:varnhmain}
\begin{aligned}[b]
\Var{N(H)}&=C_1\Exp{N(H_1)}+\dots+C_m\Exp{N(H_m)}+\Exp{N(H)}^2O(n^{-1}).
\end{aligned}
\end{equation}
for constants $C_1,\dots,C_m$. 
Using~\eqref{eq:varnhmain}, we can determine for any motif $H$ whether it is self-averaging or not. First, we find all motifs that are created by merging two copies of $H$. For the triangle motif for example, these motifs are the bow-tie, where two triangles are merged at one single vertex, the diamond of Fig.~\ref{fig:graphlet4}b, and the triangle itself.
We find the order of magnitude of the expected number of these motifs using Theorem~\ref{thm:expmotifsx} to obtain the variance of $N(H)$. We divide by $\Exp{N(H)}^2$, also obtained by Theorem~\ref{thm:expmotifsx}, and check whether this fraction is diverging or not. Table~\ref{tab:selfav4} shows for which values of $\tau\in(2,3)$ the  motifs on 3, 4 and 5 vertices are self-averaging. 
For example, the triangle turns out to be self-averaging only for $\tau\in(2,5/2)$. 

Here is a general observation that underlines the importance of the dominant motif structure:

\begin{theorem}\label{thm:selfav}
All self-averaging motifs for any $\tau\in(2,3)$ have dominant free variation structures that consist of vertices with hidden variables $\Theta(\sqrt{n})$ only.
\end{theorem}
We prove this theorem in Supplementary Material 5.
Note that the condition on the dominant motif structure is a necessary condition for being self-averaging, but it is not a sufficient one, as the triangle example shows. Table~\ref{tab:selfav4} shows the values of $\tau$ for which all connected motifs on 3, 4 and 5 vertices are self-averaging.
\begin{table}[tbp]
	\centering
			\begin{tabular}{lll}
				\toprule
				self-averaging for &subfigs of Fig.~\ref{fig:graphlet4exp} &subfigs of Fig.~\ref{fig:motif5exp} \\
				\midrule
				(2,3) & c & a,b,c,d  \\
				(2,5/2) & a & f,g,l,o\\
				(2,7/3) &  & i \\
				-     & b,d,e,f,g,h  &  e,h,j,k,m,n,p,q,r,s,t,u \\
				\bottomrule
			\end{tabular}%
		\caption{The values of $\tau\in(2,3)$ where the motifs of Figs.~\ref{fig:graphlet4exp} and~\ref{fig:motif5exp} are self-averaging.}
		\label{tab:selfav4}%
\end{table}
Combining the classification of the motifs into self-averaging and non-self-averaging with the classification based on the value of $B_f(H)$ from~\eqref{eq:maxeqx} as well as the difference between the expected and typical number of motifs yields a classification into the following three types of motifs:

\textit{Type I: Motifs with small variance.} $B_f(H)=0$ and $\Var{N(H)}/\Exp{N(H)}^2\to{0}$. These motifs only contain vertices of degrees $\Theta(\sqrt{n})$. The number of such rescaled motifs converges to a constant~\cite{hofstad2017d}. Furthermore, the variance of the number of motifs is small compared to the second moment, so that the fluctuations of these types of motifs are quite small and vanish in the large network limit. The triangle for $\tau<5/2$ is an example of such a motif, shown in Figs~\ref{fig:triangleselfavg}b and~\ref{fig:triangleselfavg}e.

 \textit{Type II: Concentrated, non-self-averaging motifs.} $B_f(H)=0$ and $\Var{N(H)}/\Exp{N(H)}^2\nrightarrow{0}$. These motifs also only contain vertices of degrees $\sqrt{n}$. Again, the rescaled number of such motifs converges to a constant in probability~\cite{hofstad2017d}. Thus, most network samples contain a similar amount of motifs as $n$ grows large, even though these motifs are non-self-averaging. Still, in rare network samples the number of motifs significantly deviates from its typical number, causing the variance of the number of motifs to be large. Figs~\ref{fig:triangleselfavg}a and~\ref{fig:triangleselfavg}d illustrate this for triangle counts for $\tau\geq 5/2$. The fluctuations are larger than for the concentrated motifs, but most of the samples have motif counts close to the expected value.
 
\textit{Type III: Non-concentrated motifs.} $B_f(H)>0$. These motifs contain hub vertices. The expected and typical number of such motifs therefore scale differently in $n$. By Theorem~\ref{thm:selfav}, these motifs are non-self-averaging. The rescaled number of such motifs may not converge to a constant, so that two network samples contain significantly different motif counts. Figs~\ref{fig:triangleselfavg}c and~\ref{fig:triangleselfavg}f show that the fluctuations of these motifs are indeed of a different nature, since most network samples have motif counts that are far from the expected value.

\begin{figure}[tb]
	\centering
	\captionsetup[subfigure]{justification=centering}
	\begin{subfigure}[t]{0.3\linewidth}
		\centering
		{ small variance}
		\includegraphics[width=0.8\textwidth]{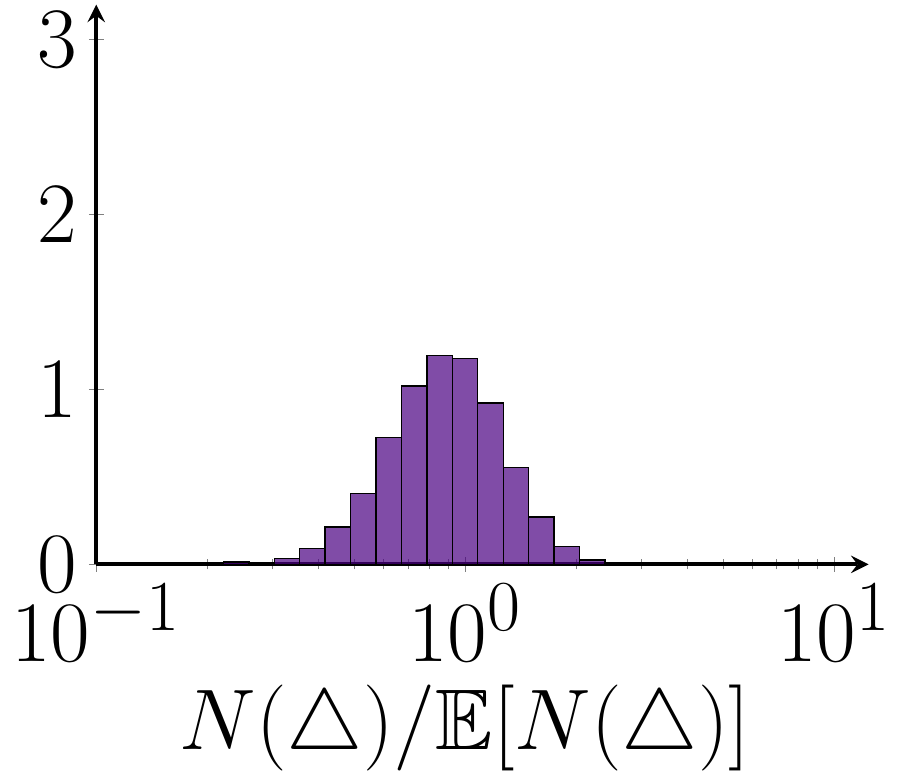}
		\vspace{-0.5\baselineskip}
		\caption{$\tau=2.2,n=10^5$, triangles}
		\label{var:fig:tr22100000}
	\end{subfigure}
	\hspace{0.2cm}
	\begin{subfigure}[t]{0.3\linewidth}
		\centering
		{concentrated, non-self-averaging}
		\includegraphics[width=0.8\textwidth]{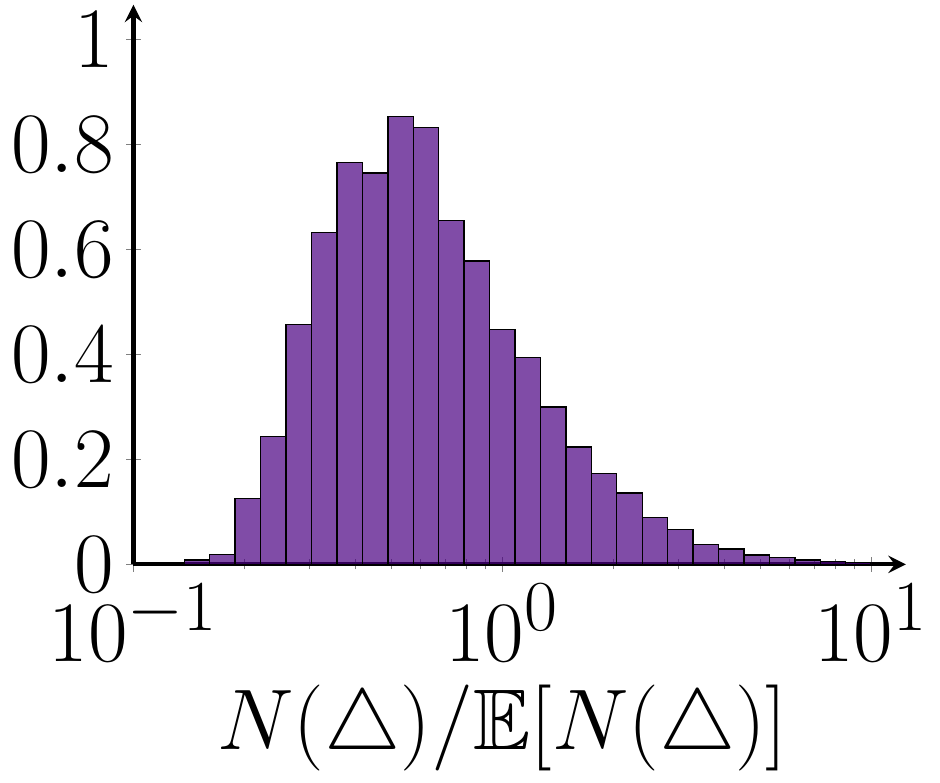}
		\vspace{-0.5\baselineskip}
		\caption{$\tau=2.8,n=10^5$, triangles}
		\label{var:fig:tr28100000}
	\end{subfigure}
	\hspace{0.2cm}
	\begin{subfigure}[t]{0.3\linewidth}
		\centering
		{non-concentrated}
		\includegraphics[width=0.8\textwidth]{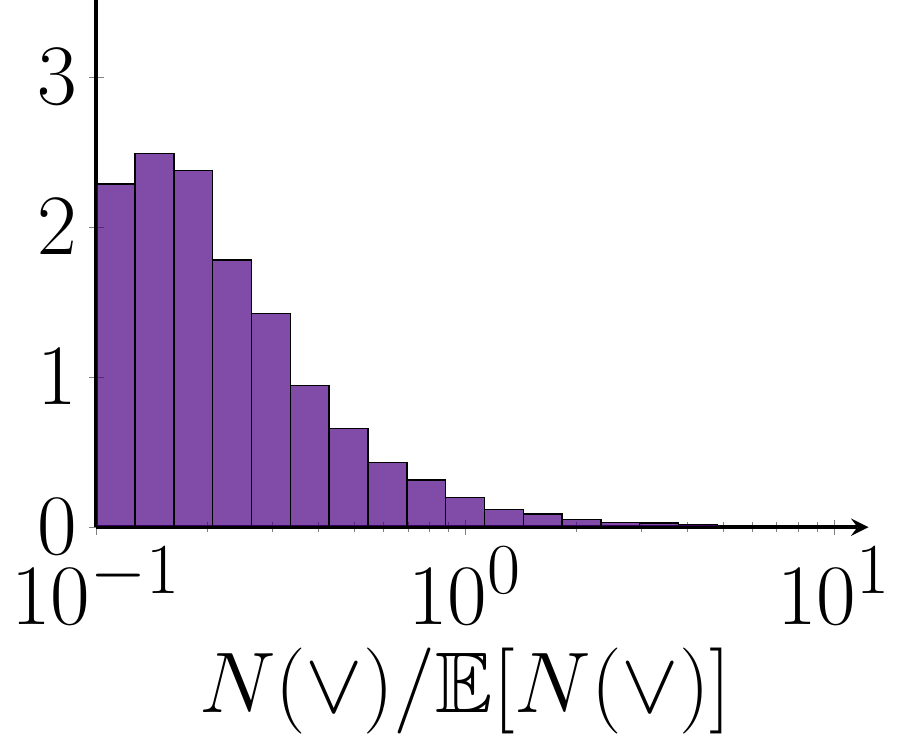}
		\vspace{-0.5\baselineskip}
		\caption{$\tau=2.2,n=10^5$, \ \  wedges}
		\label{var:fig:w22100000}
	\end{subfigure}
	
	\begin{subfigure}[t]{0.3\linewidth}
		\centering
		\includegraphics[width=0.8\textwidth]{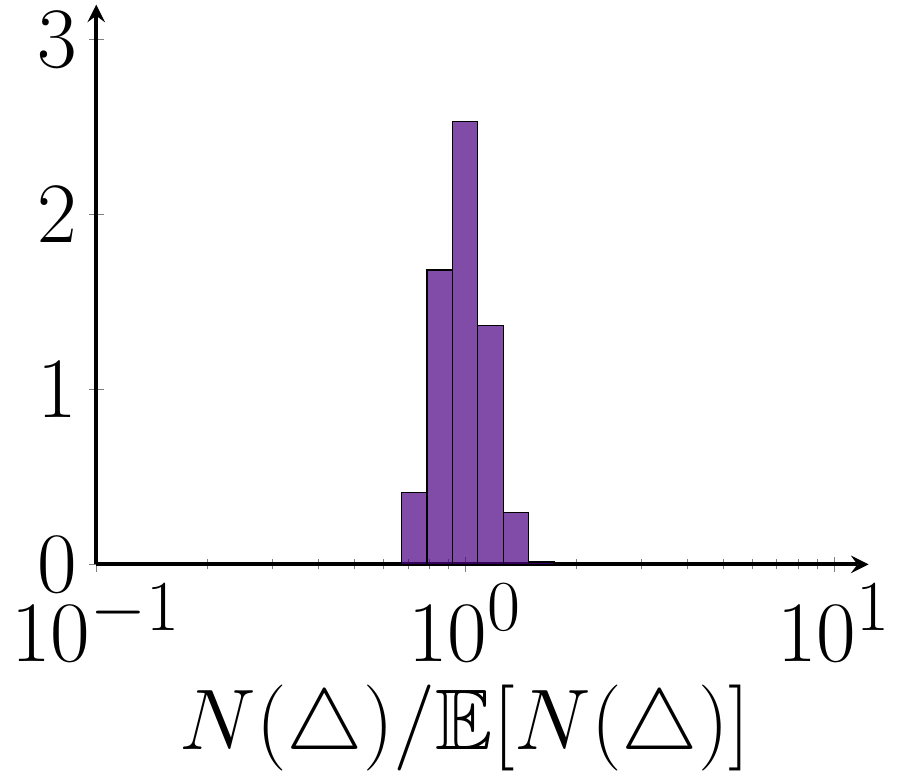}
		\vspace{-0.5\baselineskip}
		\caption{$\tau=2.2$, $ n=5\cdot 10^6$, triangles}
		\label{var:fig:tr225000000}
	\end{subfigure}
	\hspace{0.2cm}
	\begin{subfigure}[t]{0.3\linewidth}
		\centering
		\includegraphics[width=0.8\textwidth]{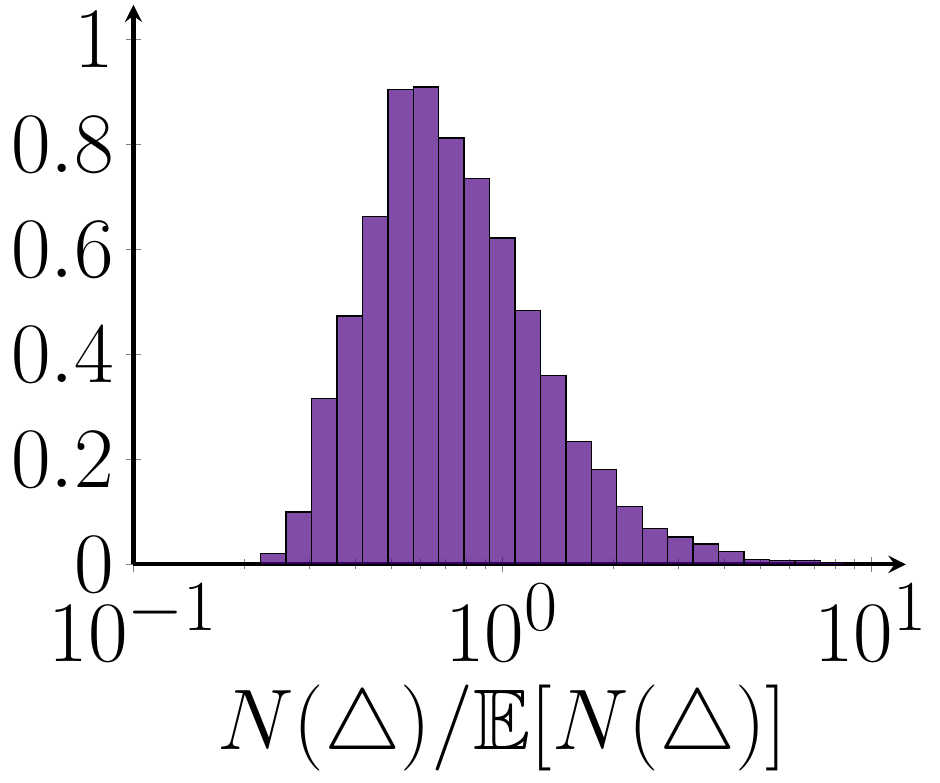}
		\vspace{-0.5\baselineskip}
		\caption{$\tau=2.8$, $n=5 \cdot 10^6$, triangles}
		\label{var:fig:tr285000000}
	\end{subfigure}
	\hspace{0.2cm}
	\begin{subfigure}[t]{0.3\linewidth}
		\centering
		\includegraphics[width=0.8\textwidth]{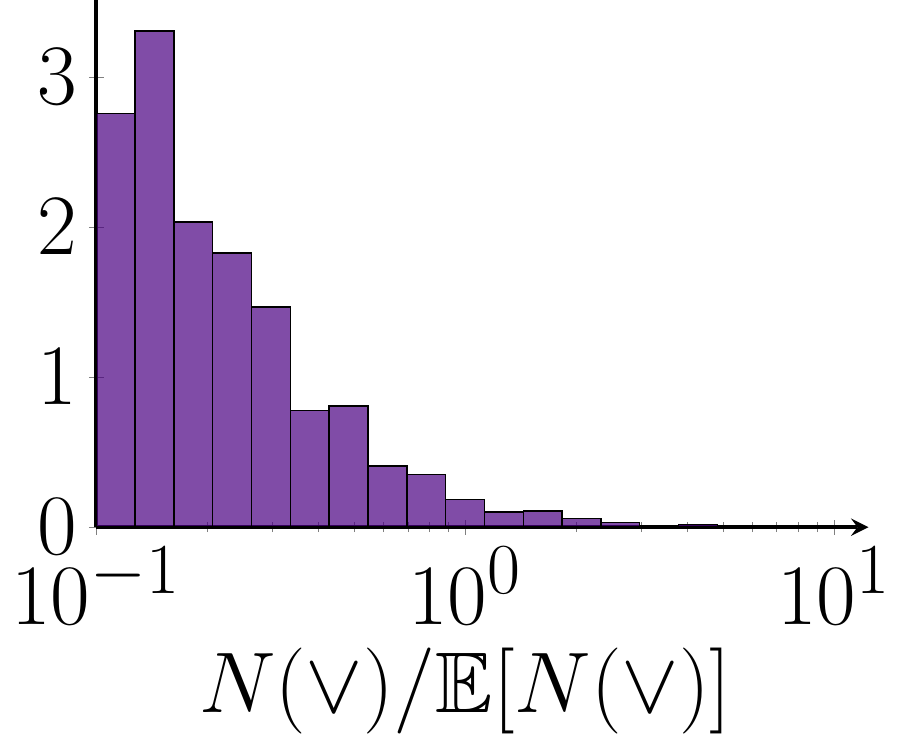}
		\vspace{-0.5\baselineskip}
		\caption{$\tau=2.2$, $n=5\cdot 10^6$, wedges}
		\label{var:fig:w225000000}
	\end{subfigure}
	\caption{Density approximation of the normalized triangle and wedge counts for various values of $\tau$ and $n$, obtained over $10^4$ network samples.}
	\label{fig:triangleselfavg}
\end{figure}

\subsection*{Data}\label{sec:data}

\begin{table}[htbp]
	\centering
		\begin{tabular}{lrrr}
			\toprule
			\textbf{} & $n$     & $m $    & $\tau$ \\
			\midrule
			\textbf{Gowalla} & 196591 & 950327 & 2.65 \\
			\textbf{Oregon} & 11174 & 23409 &  2.08\\
			\textbf{Enron} & 36692 & 183831 & 1.97 \\
			\textbf{PGP} & 10680 & 24316 &  2.24 \\
			\textbf{Hep} & 9877  & 25998 &  3.50\\
			\bottomrule
		\end{tabular}%
	\caption{Statistics of the five data sets, where $n$ is the number of vertices, $m$ the number of edges, and $\tau$ the power-law exponent fitted by the procedure of~\cite{clauset2009}.}
	\label{tab:data}%
\end{table}%

We now investigate motifs in five real-world networks with heavy-tailed degree distributions: the Gowalla social network~\cite{snap}, the Oregon autonomous systems network~\cite{snap}, the Enron email network~\cite{klimt2004,snap}, the PGP web of trust~\cite{boguna2004a} and the High Energy Physics collaboration network (HEP)~\cite{snap}.
Table~\ref{tab:data} provides detailed statistics of these data sets. 
Because the number of motifs can be obtained from the number of graphlets, we focus on graphlet counts.
Figure~\ref{fig:motifsorder} shows the graphlet counts on a logarithmic scale. The order of the graphlets is from the most occurring graphlet (the claw), to the least occurring graphlet (the square and $K_4$) in the hidden-variable model, see Supplementary Figure 2. 
In three networks the motif ordering follows that of the hidden-variable model, while in two networks the ordering is different.
In the HEP collaboration network, for example, $K_4$ occurs more frequently than the square. While this is not predicted by the hidden-variable model, it naturally arises due to the frequently occurring collaboration between four authors, which creates $K_4$ instead of the square. It would be interesting to see if this deviation from the ordering of the hidden-variable model can be linked to the specific nature of the data set in other examples.

\begin{figure}[tb]
	\centering
	\includegraphics[width=0.6\linewidth]{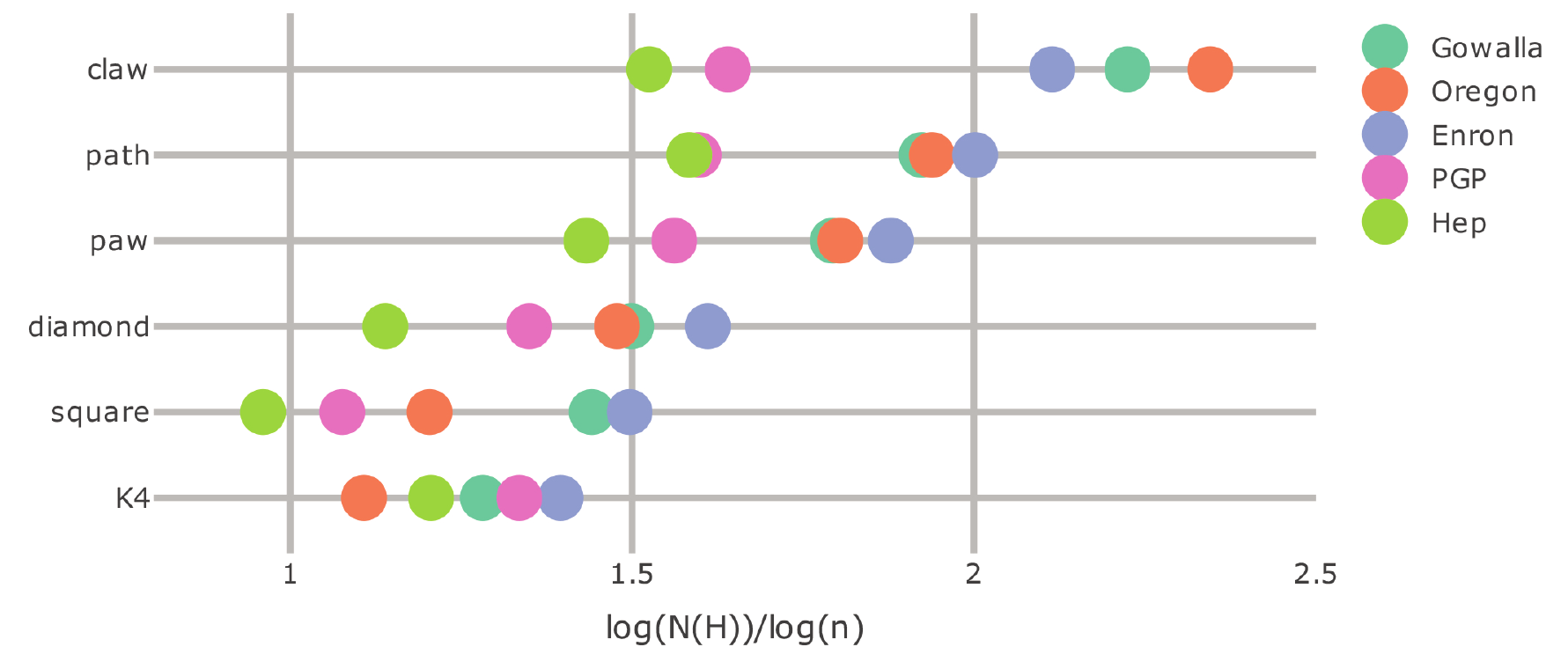}
	\caption{Number of graphlets on four vertices in five data sets on logarithmic scale: $\log(N(H))/\log(n)$. The ordering of the six different graphlets is from the most occurring in the hidden-variable model to the least. The values of the graphlet counts are presented in Supplementary Table~1.}
	\label{fig:motifsorder}
\end{figure}

\begin{figure}[tb]
	\centering
	\begin{minipage}[b]{\linewidth}
		\centering
		\includegraphics[width=0.7\linewidth]{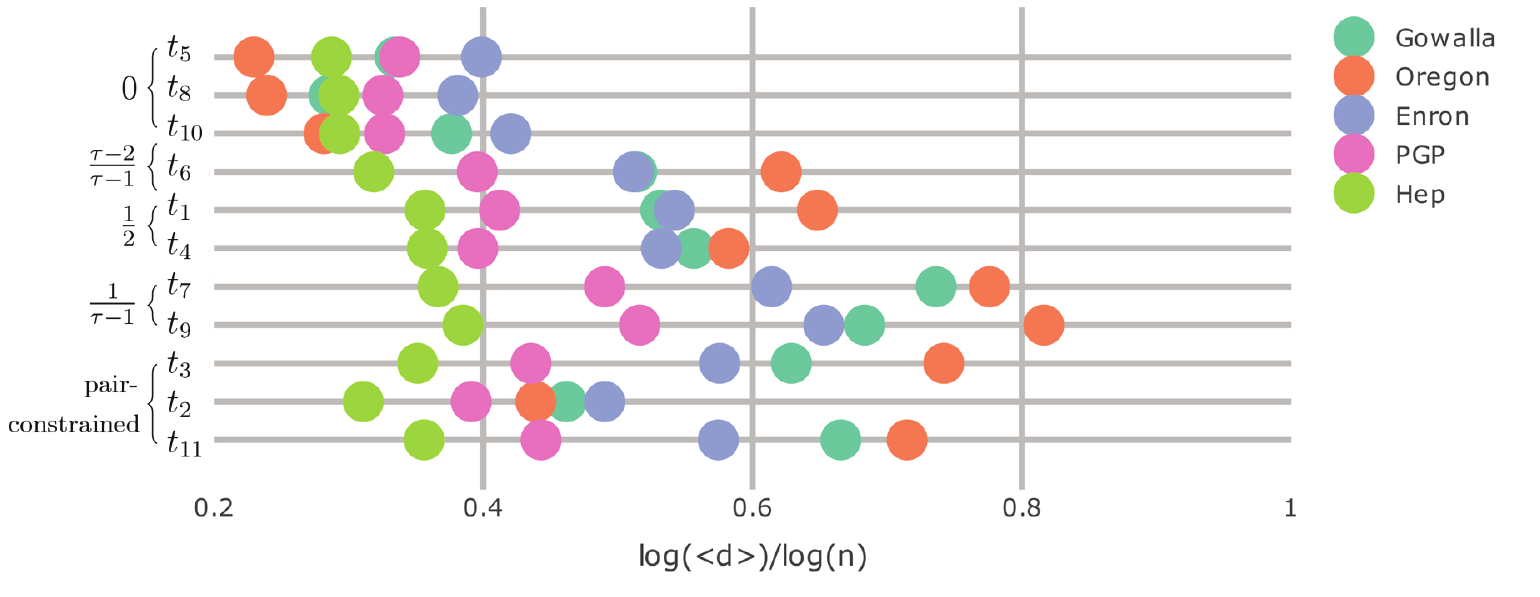}
		\caption{Average degree of the vertex types displayed in Supplementary Figure 2 in 5 data sets on logarithmic scale. The curly brackets indicate the typical degree exponent of the vertex type in the hidden-variable model. }
		\label{fig:motifsdegree}
	\end{minipage}
\end{figure}

Supplementary Figure~2 enumerates all possible vertex types in graphlets on 4 vertices. In the hidden-variable model, vertex types $t_7$ and $t_9$ have typical degrees proportional to $n^{1/(\tau-1)}$ in the hidden-variable model, vertex types $t_1$ and $t_4$ typically have degrees proportional to $\sqrt{n}$, vertex type $t_6$ typically has degree proportional to $n^{(\tau-2)/(\tau-1)}$ and vertex types $t_5,t_8,t_{10},t_{11}$ typically have constant degree. vertex types $t_2, t_3$ and $t_{11}$ do not have a unique optimizer. The degrees of these vertex types are pair-constrained (see the proof of Lemma~\ref{lem:diamond}). 
Figure~\ref{fig:motifsdegree} shows the typical degree of all 11 vertex types in the five real-world data sets. Vertices with typical degree 1 in the hidden-variable model have the lowest degree in the five data sets. Vertices that have typical degree $n^{1/(\tau-1)}$ in the hidden-variable model also have the highest degree among all vertex types in these five real-world data sets. Thus, typical degrees of vertices in a graphlet roughly follow the same ordering as in the hidden-variable model in these data sets. 

The High Energy Physics collaboration network does not have a large distinction between the degrees of the different vertex types. This may be related to the fact that this network has less heavy-tailed degrees than the other networks (see Table~\ref{tab:data}).

\section*{Discussion}\label{sec:conclusion}
By developing a variational principle for the dominant degree composition of motifs in the hidden-variable model, we have identified the asymptotic growth of motif counts and their fluctuations for all motifs. This allowed us to determine for which values of the degree exponent $\tau\in(2,3)$ the number of motifs is self-averaging. We further divide the non-self-averaging motifs into two classes with substantially different concentration properties.

Hub vertices in dominant motif structures cause wild degree fluctuations and non-self-averaging behavior, so that large differences between the average motif count and the motif count in one sample of the random network model arise.
Non-self-averaging motifs without a hub vertex show milder fluctuations. 

We expect that the variational principle can be extended to different random graph models, such as the hyperbolic random graph, the preferential attachment model and random intersection graphs. For example, for the hyperbolic random graph, the dominant structure of complete graphs is known to be $\sqrt{n}$ degrees~\cite{friedrich2015} like in the hidden-variable model, but the dominant structures of other motifs are yet unknown.

In this paper, we presented a case study for  motifs on 4 vertices in five scale-free network data sets. It would be interesting to perform larger network data experiments to investigate whether motifs in real-world network data also have typical vertex degrees, and to what extent these vertex degrees are similar to the ones of the hidden-variable model. Similarly, investigating the typical behavior of motifs on more than 4 vertices in real-world data compared to the hidden-variable model is another topic for future research.

It would also be interesting to develop statistical tests for the presence of motifs in real-world data using the results from this paper. For example, one could compare the ordering of all motifs for size $k$ from the most frequent occurring to the least frequent occurring motif, and compare this to the ordering in a hidden-variable with the same degree-exponent. This could shed some light on which motifs in a given data set appear more often than expected.

\paragraph{Acknowledgements.}
This work is supported by NWO TOP grant 613.001.451 and by the NWO Gravitation Networks grant 024.002.003.
The work of RvdH is further supported by the NWO VICI grant 639.033.806.  The work of JvL is further supported by an NWO TOP-GO grant and by an ERC Starting Grant.

\section*{Methods}
\subsection*{Fluctuations}
 \paragraph{Triangle fluctuations.}
We first illustrate how we can apply the variational principle to obtain the variance of the number of subgraphs by computing the variance of the number of triangles in the hidden-variable model.
Let $\triangle$ denote the number of triangles, and let $\triangle_{i,j,k}$ denote the event that vertices $i,j$ and $k$ form a triangle. Then, we can write the number of triangles as
\begin{equation}
\triangle=\tfrac{1}{6}\sideset{}{'}\sum_{i,j,k\in[n]}\mathbbm{1}_{\triangle_{i,j,k}},
\end{equation}
where $\sum^{'}$ denotes the sum over distinct indices. 
Thus, the variance of the number of triangles can be written as
\begin{align}\label{eq:var}
\Var{\triangle}& =\sideset{}{'}\sum_{i,j,k\in[n]}\sideset{}{'}\sum_{s,t,u\in[n]}\Prob{\triangle_{i,j,k},\triangle_{s,t,u}}-\Prob{\triangle_{i,j,k}}\Prob{\triangle_{s,t,u}}.
\end{align}
When $i,j,k$ and $s,t,u$ do not overlap, the hidden variables of $i,j,k$ and $s,t,u$ are independent, so that the event that $i,j$ and $k$ form a triangle and the event that $s,t$ and $u$ form a triangle are independent. Thus, when $i,j,k,s,t,u$ are all distinct, $\Prob{\triangle_{i,j,k},\triangle_{s,t,u}}=\Prob{\triangle_{i,j,k}}\Prob{\triangle_{s,t,u}}$, so that the contribution from 6 distinct indices to~\eqref{eq:var} is zero. On the other hand, when $i=u$ for example, the first term in~\eqref{eq:var} denotes the probability that a bow tie (see Fig.~\ref{fig:motif5exp}i) is present with $i$ as middle vertex.
Furthermore, since the degrees are i.i.d.\ and the edge statuses are independent as well, $\Prob{\triangle_{i,j,k}}$ is the same for any $i\neq j\neq k$, so that
\begin{equation}
\Prob{\triangle_{i,j,k}}=\frac{\Exp{\triangle}}{6\binom{n}{3}} = \frac{\Exp{\triangle}}{6 n^3}(1+o(1)).
\end{equation}
This results in
\begin{equation}
\begin{aligned}[b]\label{eq:vartriang}
\Var{\triangle} 
& =9\Exp{\#\text{ bow-ties}}- 9n^{-1}\Exp{\triangle}^2+18\Exp{\#\text{ diamonds}}-18n^{-2}\Exp{\triangle}^2+6\Exp{\triangle}-6n^{-3}\Exp{\triangle}^2\\
& =  9\Exp{\text{\# bow-ties}}+18\Exp{\text{\# diamonds}} +6\Exp{\triangle}+\Exp{\triangle}^2O(n^{-1}),
\end{aligned}
\end{equation}
where the diamond motif is as in Fig.~\ref{fig:graphlet4}b. 
The combinatorial factors 9,18 and 6 arise because there are 9 ways to construct a bow tie (18 for a diamond, and 6 for a triangle) by letting two triangles overlap. The diamond motif does not satisfy the assumption in Theorem~\ref{thm:expmotifsx} that the optimal solution to~\eqref{eq:maxeqx} is unique. However, we can show the following result:

\begin{lemma}\label{lem:diamond}
	$\Exp{{\rm number \  of \  diamonds}}=\thetao{n^{6-2\tau}}\log(n)$.
\end{lemma}
\begin{proof}
	Let $i$ and $j$ be the vertices at the diagonal of the diamond, and $k$ and $s$ the corner vertices. Then~\eqref{eq:maxalphx} is optimized for $\alpha_i=\beta$, $\alpha_j=1-\beta$, $\alpha_k=\beta$ and $\alpha_s=1-\beta$ for all values of $\beta\in[1/2,1]$ (see Supplementary Information~6). All these optimizers together give the major contribution to the number of diamonds.  Thus, we need to find the number of sets of four vertices, satisfying
	\begin{equation}\label{eq:diamonddeg}
	h_ih_j=\Theta(n), \quad h_i>h_j,\quad \quad h_k=\Theta(h_i),\quad h_s=\Theta(h_j).
	\end{equation}
	Given $h_i$ and $h_j$, the number of sets of two vertices $k,s$ with $h_k=\Theta(h_i)$ and $h_s=\Theta(h_j)$ is given by $n^{2}h_i^{1-\tau}h_j^{1-\tau}=\Theta(n^{3-\tau})$, where we used that $h_ih_j=\Theta(n)$. The number of sets of vertices $i,j$ such that $h_ih_j=\Theta(n)$ can be found using that the product of two independent power-law random variables is again distributed as a power law, with an additional logarithmic term~\cite[Eq. (2.16)]{hofstad2014} (where in our setting equality holds in~\cite[Eq. (2.16)]{hofstad2014}, since we assume a pure power-law distribution). Thus, the number of sets of vertices with $h_ih_j=\Theta(n)$ scales as $n^{2}n^{1-\tau}\log(n)$. Then, the expected number of sets of four vertices satisfying all constraints on the degrees scales as $n^{6-2\tau}\log(n)$. By~\eqref{eq:phsubx}, the probability that a diamond exists on degrees satisfying~\eqref{eq:diamonddeg} is asymptotically constant, so that the expected number of diamonds also scales as $n^{6-2\tau}\log(n)$. 
\end{proof}
Theorem~\ref{thm:expmotifsx} gives for the number of bow ties that (using~\cite{stegehuiscode2019} to find the optimal partition)
\begin{equation}
\Exp{\#\text{ bow ties}}= \begin{cases}
\Theta(n^{\frac{5}{2}(3-\tau)}) & \tau<\tfrac{7}{3},\\
\Theta(n^{4-\tau}) & \tau\geq \tfrac{7}{3},
\end{cases}
\end{equation}
and for the number of triangles (again using~\cite{stegehuiscode2019}) that $\Exp{\triangle}=\Theta(n^{3(3-\tau)/2})$.
Combining this with~\eqref{eq:vartriang} results in
\begin{equation}
\Var{\triangle}= \begin{cases}
\Theta(n^{\frac{5}{2}(3-\tau)}) & \tau<\tfrac{7}{3},\\
\Theta(n^{4-\tau}) & \tau\geq \tfrac{7}{3}.
\end{cases}
\end{equation}
To investigate whether the triangle motif is self-averaging, we need to compare the variance to the second moment of the number of triangles, which results in
\begin{equation}
\frac{\Var{\triangle}}{\Exp{\triangle}^2}=\begin{cases}
\Theta(n^{\frac{1}{2}(\tau-3)}), & \tau<\tfrac{7}{3},\\
\Theta(n^{2\tau-5}), & \tau\geq \tfrac{7}{3}.
\end{cases}
\end{equation}
Therefore,
\begin{equation}\label{eq:triangselfavglimit}
\lim_{n\to\infty}\frac{\Var{\triangle}}{\Exp{\triangle}^2}=\begin{cases}
0 & \tau<\tfrac{5}{2},\\
\infty & \tau >\tfrac{5}{2}.
\end{cases}
\end{equation}
For $\tau=5/2$, the limit in~\eqref{eq:triangselfavglimit} is of constant order of magnitude. 
Thus, the number of triangles is self-averaging as long as $\tau<\tfrac{5}{2}$. When $\tau\geq\tfrac{5}{2}$ the number of triangles is not self-averaging.

\paragraph{General motif fluctuations.}\label{sec:vareq}
We now compute the variance of general motifs, similar to the triangle example. 
Let $\boldsymbol{i}=(i_1,\dots,i_k)$ be such that $i_p\neq i_q$ when $p\neq q$. We can then write the variance as
\begin{equation}\label{eq:varproof}
\begin{aligned}[b]
\Var{N(H)}&=\sum_{\boldsymbol{i}\in [n]^k}\sum_{\boldsymbol{j}\in [n]^k}\Big(\Prob{H_{\boldsymbol{i}},H_{\boldsymbol{j}} \text{ present}}-\Prob{H_{\boldsymbol{i}} \text{ present}}\Prob{H_{\boldsymbol{j}} \text{ present}}\Big).
\end{aligned}
\end{equation}
The sum splits into several cases, depending on the overlap of $\boldsymbol{i}$ and $\boldsymbol{j}$. 
The term where $\boldsymbol{i}$ and $\boldsymbol{j}$ do not overlap 
equals zero, since edges between vertices that do not overlap are independent. 

Now suppose $\boldsymbol{i}$ and $\boldsymbol{j}$ overlap at $i_{t_1},\dots i_{t_r}$ and $j_{s_1},\dots ,j_{s_r}$ for some $r>0$. Then $\Prob{H_{\boldsymbol{i}},H_{\boldsymbol{j}} \text{ present}}$ is equal to the probability that motif $\tilde{H}$ is present on vertices $i_1,\dots,i_k,j_1,\dots,j_k\setminus j_{s_1},\dots ,j_{s_r}$, where $\tilde{H}$ denotes the motif that is constructed by merging two copies of $H$ at $i_{t_1}$ with $j_{s_1}$, at $i_{t_2}$ with $j_{s_2}$ and so on. Thus, this term can be written as
\begin{equation}
\sideset{}{'}{\sum}_{t_1,\dots,t_{2k-r}}\Prob{\tilde{H}_{t_1,\dots,t_{2k-r}}\text{ present}}=\Exp{N(\tilde{H})},
\end{equation}
where $\sideset{}{'}{\sum}$ denotes a sum over distinct indices. 
Furthermore, since the degrees are i.i.d.\ as well as the connection probabilities, $\Prob{H_{\boldsymbol{i}}\text{ present}}=\Exp{N(H)}/{n\choose k}$. Thus, 
\begin{equation}
\begin{aligned}[b]
& \sideset{}{'}{\sum}_{\mathclap{t_1,\dots,t_{2k-r}}}\Prob{H_{t_1,\dots,t_{k}}\text{ present}}\Prob{H_{t_{k-r},\dots,t_{2k-r}}\text{ present}}=n^{-r}\Exp{N(H)}^2O(1).
\end{aligned}
\end{equation}
Let $H_1,\dots,H_l$ denote all motifs that can be constructed by merging two copies of $H$ at at least one vertex. We can then write the variance of the motif count as (see \cite{frank1979,picard2008,matias2006,ostilli2014})
\begin{equation}\label{eq:varnh}
\begin{aligned}[b]
\Var{N(H)}&=C_1\Exp{N(H_1)}+\dots+C_l\Exp{N(H_l)}+\Exp{N(H)}^2O(n^{-1}).
\end{aligned}
\end{equation}
where $C_i$ is a combinatorial constant that denotes the number of distinct ways to merge two copies of $H$ into $H_i$. These constants satisfy~\cite{frank1979}
\begin{equation}
\sum_{i=1}^{l}C_i=\sum_{s=0}^{k-1}{ \binom{k}{s}}^2(k-s)!.
\end{equation}

\bibliographystyle{unsrt}
\bibliography{../references}

\newpage

\beginsupplementnote

\section{Additional motif structures}\label{sec:motiftyp5}
Figure~\ref{fig:motif5} presents the typical dominant motif structures for all connected motifs on 5 vertices.

\begin{figure*}[htb]
	\centering
	\begin{subfigure}[t]{0.13\linewidth}
		\centering
		\begin{tikzpicture}
		\tikzstyle{edge} = [draw,thick,-]
		\node[S3] (a) at (90:0.8) {};
		\node[S3] (b) at (162:0.8) {};
		\node[S3] (c) at (234:0.8) {};
		\node[S3] (d) at (306:0.8) {};
		\node[S3] (e) at (378:0.8) {};
		\draw[edge] (a)--(b);
		\draw[edge] (c)--(b);
		\draw[edge] (d)--(b);
		\draw[edge] (a)--(c);
		\draw[edge] (a)--(d);
		\draw[edge] (c)--(e);
		\draw[edge] (e)--(d);
		\draw[edge] (b)--(e);
		\draw[edge] (a)--(e);
		\draw[edge] (c)--(d);
		\end{tikzpicture}	
		\caption{$n^{\frac{5}{2}(3-\tau)}$}
	\end{subfigure}
	\begin{subfigure}[t]{0.13\linewidth}
		\centering
		\begin{tikzpicture}
		\tikzstyle{edge} = [draw,thick,-]
		\node[S3] (a) at (90:0.8) {};
		\node[S3] (b) at (162:0.8) {};
		\node[S3] (c) at (234:0.8) {};
		\node[S3] (d) at (306:0.8) {};
		\node[S3] (e) at (378:0.8) {};
		\draw[edge] (a)--(b);
		\draw[edge] (c)--(b);
		\draw[edge] (d)--(b);
		\draw[edge] (a)--(c);
		\draw[edge] (a)--(d);
		\draw[edge] (c)--(e);
		\draw[edge] (e)--(d);
		\draw[edge] (b)--(e);
		\draw[edge] (a)--(e);
		\end{tikzpicture}	
		\caption{$n^{\frac{5}{2}(3-\tau)}$}
	\end{subfigure}
	\begin{subfigure}[t]{0.13\linewidth}
		\centering
		\begin{tikzpicture}
		\tikzstyle{edge} = [draw,thick,-]
		\node[S3] (a) at (0,0) {};
		\node[S3] (b) at (1,0) {};
		\node[S3] (c) at (0,1) {};
		\node[S3] (d) at (1,1) {};
		\node[S3] (e) at (0.5,1.5) {};
		\draw[edge] (a)--(b);
		\draw[edge] (c)--(b);
		\draw[edge] (d)--(b);
		\draw[edge] (a)--(c);
		\draw[edge] (a)--(d);
		\draw[edge] (c)--(e);
		\draw[edge] (d)--(e);
		\draw[edge] (c)--(d);
		\end{tikzpicture}	
		\caption{$n^{\frac{5}{2}(3-\tau)}$}
	\end{subfigure}
	\begin{subfigure}[t]{0.13\linewidth}
		\centering
		\begin{tikzpicture}
		\tikzstyle{edge} = [draw,thick,-]
		\node[S3] (a) at (0,0) {};
		\node[S3] (b) at (1,0) {};
		\node[S3] (c) at (0,1) {};
		\node[S3] (d) at (1,1) {};
		\node[S3] (e) at (0.5,0.5) {};
		\draw[edge] (a)--(b);
		\draw[edge] (c)--(a);
		\draw[edge] (b)--(e);
		\draw[edge] (a)--(e);
		\draw[edge] (b)--(d);
		\draw[edge] (c)--(e);
		\draw[edge] (d)--(e);
		\draw[edge] (c)--(d);
		\end{tikzpicture}	
		\caption{$n^{\frac{5}{2}(3-\tau)}$}
	\end{subfigure}
	\begin{subfigure}[t]{0.13\linewidth}
		\centering
		\begin{tikzpicture}
		\tikzstyle{edge} = [draw,thick,-]
		\node[S3] (a) at (0,0) {};
		\node[S3] (b) at (1,0) {};
		\node[S3] (c) at (0,1) {};
		\node[S3] (d) at (1,1) {};
		\node[n1] (e) at (0.5,1.5) {};
		\draw[edge] (a)--(b);
		\draw[edge] (c)--(b);
		\draw[edge] (d)--(b);
		\draw[edge] (a)--(c);
		\draw[edge] (a)--(d);
		\draw[edge] (c)--(e);
		\draw[edge] (c)--(d);
		\end{tikzpicture}	
		\caption{$n^{\frac{13}{2}-2\tau}$}
	\end{subfigure}
	\begin{subfigure}[t]{0.13\linewidth}
		\centering
		\begin{tikzpicture}
		\tikzstyle{edge} = [draw,thick,-]
		\node[S3] (a) at (0,0) {};
		\node[S3] (b) at (1,0) {};
		\node[S3] (c) at (0,1) {};
		\node[S3] (d) at (1,1) {};
		\node[S3] (e) at (0.5,0.5) {};
		\draw[edge] (a)--(b);
		\draw[edge] (c)--(a);
		\draw[edge] (b)--(e);
		\draw[edge] (b)--(d);
		\draw[edge] (c)--(e);
		\draw[edge] (c)--(d);
		\draw[edge] (e)--(d);
		\end{tikzpicture}	
		\caption{$n^{\frac{5}{2}(3-\tau)}$}
	\end{subfigure}	
	\begin{subfigure}[t]{0.13\linewidth}
		\centering
		\begin{tikzpicture}
		\tikzstyle{edge} = [draw,thick,-]
		\node[S3] (a) at (0,0) {};
		\node[S3] (b) at (1,0) {};
		\node[S3] (c) at (0,1) {};
		\node[S3] (d) at (1,1) {};
		\node[S3] (e) at (0.5,1.5) {};
		\draw[edge] (a)--(b);
		\draw[edge] (c)--(b);
		\draw[edge] (d)--(b);
		\draw[edge] (a)--(c);
		\draw[edge] (c)--(e);
		\draw[edge] (d)--(e);
		\draw[edge] (c)--(d);
		\end{tikzpicture}	
		\caption{$n^{\frac{5}{2}(3-\tau)}$}
	\end{subfigure}
	
	\begin{subfigure}[t]{0.13\linewidth}
		\centering
		\begin{tikzpicture}
		\tikzstyle{edge} = [draw,thick,-]
		\node[S1] (a) at (0,0) {};
		\node[S1] (b) at (1,0) {};
		\node[S2m] (c) at (0,1) {};
		\node[S2m] (d) at (1,1) {};
		\node[S1] (e) at (0.5,1.5) {};
		\draw[edge] (d)--(b);
		\draw[edge] (c)--(b);
		\draw[edge] (d)--(b);
		\draw[edge] (a)--(c);
		\draw[edge] (a)--(d);
		\draw[edge] (c)--(e);
		\draw[edge] (d)--(e);
		\draw[edge] (c)--(d);
		\end{tikzpicture}	
		\caption{$n^{9-3\tau}$}
		\label{mot:fig:K23}
	\end{subfigure}
	\begin{subfigure}[t]{0.13\linewidth}
		\centering
		\begin{tikzpicture}
		\tikzstyle{edge} = [draw,thick,-]
		\node (a) at (0,0) {};
		\node (b) at (1,0) {};
		\node (c) at (0,1) {};
		\node (d) at (1,1) {};
		\node (e) at (0.5,0.5) {};
		\draw[edge] (a)--(b);
		\draw[edge] (e)--(a);
		\draw[edge] (b)--(e);
		\draw[edge] (e)--(d);
		\draw[edge] (c)--(e);
		\draw[edge] (c)--(d);
		\end{tikzpicture}	
		\caption{depends on $\tau$}
		\label{mot:fig:bowtie}
	\end{subfigure}
	\begin{subfigure}[t]{0.13\linewidth}
		\centering
		\begin{tikzpicture}
		\tikzstyle{edge} = [draw,thick,-]
		\node[S3] (a) at (0,0) {};
		\node[S3] (b) at (1,0) {};
		\node[S3] (c) at (0,1) {};
		\node[S3] (d) at (1,1) {};
		\node[n1] (e) at (0.5,1.5) {};
		\draw[edge] (a)--(b);
		\draw[edge] (d)--(b);
		\draw[edge] (a)--(c);
		\draw[edge] (a)--(d);
		\draw[edge] (c)--(e);
		\draw[edge] (c)--(d);
		\end{tikzpicture}	
		\caption{$n^{\frac{13}{2}-2\tau}$}
	\end{subfigure}
	\begin{subfigure}[t]{0.13\linewidth}
		\centering
		\begin{tikzpicture}
		\tikzstyle{edge} = [draw,thick,-]
		\node[S1] (a) at (0,0) {};
		\node[S2m] (b) at (1,0) {};
		\node[S2m] (c) at (0,1) {};
		\node[S1] (d) at (1,1) {};
		\node[n1] (e) at (0.5,1.5) {};
		\draw[edge] (a)--(b);
		\draw[edge] (d)--(b);
		\draw[edge] (a)--(c);
		\draw[edge] (b)--(c);
		\draw[edge] (c)--(e);
		\draw[edge] (c)--(d);
		\end{tikzpicture}	
		\caption{$n^{6-2\tau-\frac{1}{\tau-1}}$}
	\end{subfigure}
	\begin{subfigure}[t]{0.13\linewidth}
		\centering
		\begin{tikzpicture}
		\tikzstyle{edge} = [draw,thick,-]
		\node[S3] (a) at (0,0) {};
		\node[S3] (b) at (1,0) {};
		\node[S3] (c) at (0,1) {};
		\node[S3] (d) at (1,1) {};
		\node[S3] (e) at (0.5,1.5) {};
		\draw[edge] (a)--(b);
		\draw[edge] (d)--(b);
		\draw[edge] (a)--(c);
		\draw[edge] (e)--(d);
		\draw[edge] (c)--(e);
		\draw[edge] (c)--(d);
		\end{tikzpicture}	
		\caption{$n^{\frac{5}{2}(3-\tau)}$}
	\end{subfigure}
	\begin{subfigure}[t]{0.13\linewidth}
		\centering
		\begin{tikzpicture}
		\tikzstyle{edge} = [draw,thick,-]
		\node[S1] (a) at (0,0) {};
		\node[S2m] (b) at (1,0) {};
		\node[S2m] (c) at (0,1) {};
		\node[S1] (d) at (1,1) {};
		\node[S1] (e) at (0.5,0.5) {};
		\draw[edge] (a)--(b);
		\draw[edge] (c)--(a);
		\draw[edge] (b)--(e);
		\draw[edge] (b)--(d);
		\draw[edge] (c)--(e);
		\draw[edge] (c)--(d);
		\end{tikzpicture}	
		\caption{$n^{9-3\tau}$}
		\label{mot:fig:m5dom}
	\end{subfigure}
	\begin{subfigure}[t]{0.13\linewidth}
		\centering
		\begin{tikzpicture}
		\tikzstyle{edge} = [draw,thick,-]
		\node (a) at (0,0) {};
		\node (b) at (1,0) {};
		\node (c) at (0.5,1) {};
		\node[n1] (d) at (0.5,1.5) {};
		\node (e) at (0.5,0.5) {};
		\draw[edge] (a)--(b);
		\draw[edge] (e)--(a);
		\draw[edge] (b)--(e);
		\draw[edge] (e)--(c);
		\draw[edge] (c)--(d);
		\end{tikzpicture}	
		\caption{depends on $\tau$}
		\label{mot:fig:dependtau}
	\end{subfigure}
	
	\begin{subfigure}[t]{0.13\linewidth}
		\centering
		\begin{tikzpicture}
		\tikzstyle{edge} = [draw,thick,-]
		\node[S3] (a) at (90:0.8) {};
		\node[S3] (b) at (162:0.8) {};
		\node[S3] (c) at (234:0.8) {};
		\node[S3] (d) at (306:0.8) {};
		\node[S3] (e) at (378:0.8) {};
		\draw[edge] (a)--(b);
		\draw[edge] (c)--(b);
		\draw[edge] (d)--(c);
		\draw[edge] (a)--(e);
		\draw[edge] (d)--(e);
		\end{tikzpicture}	
		\caption{$n^{\frac{5}{2}(3-\tau)}$}
	\end{subfigure}
	\begin{subfigure}[t]{0.13\linewidth}
		\centering
		\begin{tikzpicture}
		\tikzstyle{edge} = [draw,thick,-]
		\node[S1] (a) at (0,0) {};
		\node[S1] (b) at (1,0) {};
		\node[n1] (c) at (0,1) {};
		\node[n1] (d) at (1,1) {};
		\node[S2m] (e) at (0.5,0.5) {};
		\draw[edge] (a)--(b);
		\draw[edge] (e)--(a);
		\draw[edge] (b)--(e);
		\draw[edge] (e)--(d);
		\draw[edge] (c)--(e);
		\end{tikzpicture}	
		\caption{$n^{7-2\tau}$}
	\end{subfigure}
	\begin{subfigure}[t]{0.13\linewidth}
		\centering
		\begin{tikzpicture}
		\tikzstyle{edge} = [draw,thick,-]
		\node[S2m] (a) at (0,0) {};
		\node[S2m] (b) at (1,0) {};
		\node[n1] (c) at (0,1) {};
		\node[n1] (d) at (1,1) {};
		\node[S1] (e) at (0.5,0.5) {};
		\draw[edge] (a)--(b);
		\draw[edge] (e)--(a);
		\draw[edge] (b)--(e);
		\draw[edge] (b)--(d);
		\draw[edge] (c)--(a);
		\end{tikzpicture}	
		\caption{$n^{3-\tau+\frac{2}{\tau-1}}$}
	\end{subfigure}
	\begin{subfigure}[t]{0.13\linewidth}
		\centering
		\begin{tikzpicture}
		\tikzstyle{edge} = [draw,thick,-]
		\node[S1] (a) at (0,0) {};
		\node[S2m] (b) at (1,0) {};
		\node[S2m] (c) at (0,1) {};
		\node[S1] (d) at (1,1) {};
		\node[n1] (e) at (0.5,1.5) {};
		\draw[edge] (a)--(b);
		\draw[edge] (d)--(b);
		\draw[edge] (a)--(c);
		\draw[edge] (c)--(e);
		\draw[edge] (c)--(d);
		\end{tikzpicture}	
		\caption{$n^{6-2\tau+\frac{1}{\tau-1}}$}
	\end{subfigure}
	\begin{subfigure}[t]{0.13\linewidth}
		\centering
		\begin{tikzpicture}
		\tikzstyle{edge} = [draw,thick,-]
		\node[S3] (a) at (0,0) {};
		\node[S3] (b) at (1,0) {};
		\node[S3] (c) at (0,1) {};
		\node[n1] (d) at (1,1) {};
		\node[n1] (e) at (0.5,1.5) {};
		\draw[edge] (a)--(b);
		\draw[edge] (d)--(b);
		\draw[edge] (a)--(c);
		\draw[edge] (c)--(e);
		\end{tikzpicture}	
		\caption{$n^{\frac{11-3\tau}{2}}$}
	\end{subfigure}
	\begin{subfigure}[t]{0.13\linewidth}
		\centering
		\begin{tikzpicture}
		\tikzstyle{edge} = [draw,thick,-]
		\node[n1] (a) at (0,0) {};
		\node[n1] (b) at (1,0) {};
		\node[S2m] (c) at (0,1) {};
		\node[S1] (d) at (1,1) {};
		\node[n1] (e) at (0.5,1.5) {};
		\draw[edge] (d)--(b);
		\draw[edge] (a)--(c);
		\draw[edge] (c)--(e);
		\draw[edge] (c)--(d);
		\end{tikzpicture}	
		\caption{$n^{4-\tau+\frac{1}{\tau-1}}$}
	\end{subfigure}
	\begin{subfigure}[t]{0.13\linewidth}
		\centering
		\begin{tikzpicture}
		\tikzstyle{edge} = [draw,thick,-]
		\node[n1] (a) at (0,0) {};
		\node[n1] (b) at (1,0) {};
		\node[n1] (c) at (0,1) {};
		\node[n1] (d) at (1,1) {};
		\node[S2m] (e) at (0.5,0.5) {};
		\draw[edge] (e)--(a);
		\draw[edge] (b)--(e);
		\draw[edge] (e)--(d);
		\draw[edge] (c)--(e);
		\end{tikzpicture}	
		\caption{$n^{\frac{4}{\tau-1}}$}
	\end{subfigure}
	
	\vspace{-0.2cm}
	\begin{subfigure}{\linewidth}
		\centering
		\begin{tikzpicture}
		\node[S2m,label={[label distance=0.05cm]0:$n^{1/(\tau-1)}$}] (a) at (5,0) {};
		\node[S3,label={[label distance=0cm]0:$\sqrt{n}$}] (b) at (3.6,0) {};
		\node[S1,label={[label distance=0.05cm]0:$n^{(\tau-2)/(\tau-1)}$}] (c) at (1,0) {};
		\node[n1,label={[label distance=0.05cm]0:$1$}] (c) at (0,0) {};
		\node[label={[label distance=0.05cm]0:non-unique}] (d) at (7,0) {};
		\end{tikzpicture}
	\end{subfigure}
	\vspace{-0.8cm}
	\caption{Typical magnitude of the number of motifs on 5 vertices. The vertex color indicates the dominant vertex degree.}
	\label{fig:motif5}
\end{figure*}
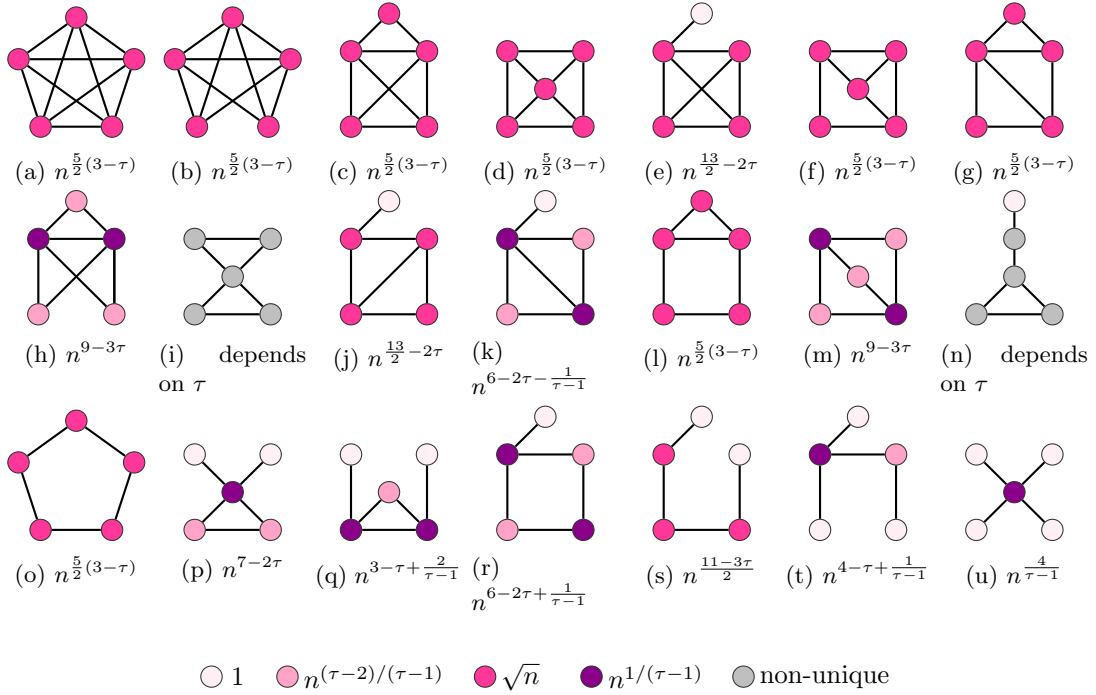

\section{Proof of Theorem \ref{thm:expmotifsx}}\label{proofthm:expmotifsx}
We now investigate the relation between the expected number of motifs and the optimization problem 
\begin{equation}\label{eq:maxeqsup}
B_f(H)=\max_{\mathcal{P}}\left[\abs{S_1}-\abs{S_2}-\frac{2E_{S_1}+E_{S_1,S_3}}{\tau-1}\right].
\end{equation}
Let $N(H,\boldsymbol{\alpha},\varepsilon)$ denote the number of times motif $H$ occurs on vertices with degrees $[\varepsilon,1/\varepsilon](n^{\alpha_i})_{i\in[k]}$ and let $\boldsymbol{\alpha}^*$ be defined as
\begin{equation}
\alpha^*_i=\begin{cases}
0 & \text {if }i\in S_1,\\
1 & \text{if }i\in S_2,\\
\tfrac{1}{2} & \text {if }i\in S_3.
\end{cases}
\end{equation}
Then, the following theorem provides a more detailed version of Theorem~\ref{thm:expmotifsx}:

\begin{theorem}[General motifs, expectation]\label{thm:expmotifs}
	Let $H$ be a motif on $k$ vertices such that the solution to~\eqref{eq:maxeqsup} is unique. Then, for any $\boldsymbol{\alpha}\neq \boldsymbol{\alpha}^*$ and $0<\varepsilon<1$,
	\begin{equation}
	\frac{\Exp{N(H,\boldsymbol{\alpha},\varepsilon)}}{\Exp{N(H,\boldsymbol{\alpha}^*,\varepsilon)}}\to 0.
	\end{equation}
	Furthermore, 
	\begin{equation}
	\frac{\Exp{N(H,\boldsymbol{\alpha}^*,\varepsilon)}}{n^{\frac{3-\tau}{2}k+\frac{\tau-2}{2}B(H)}} = f(\varepsilon)\Theta(1)
	\end{equation}
	for some function $f(\varepsilon)$ not depending on $n$. 
\end{theorem}

We now prove Theorem~\ref{thm:expmotifs}.
In the hidden-variable model, the connection probability of vertices with hidden variables $h_i$ and $h_j$ equals
\begin{equation}\label{eq:pij}
p(h_i,h_j)=\min(h_ih_j/(\mu n),1).
\end{equation}
Assume that $h_i=\Theta(n^{\alpha_i})$ for some $\alpha_i\geq 0$ for all $i$.
Then, the probability that motif $H$ exists on vertices $(v_i)_{i\in[k]}$ with hidden variables $\boldsymbol{h} = (n^{\alpha_i})_{i\in[k]}$ satisfies
\begin{equation}\label{eq:phsub}
\begin{aligned}[b]
&\Prob{H \text{ present on }(v_1,\cdots,v_k)\mid \boldsymbol{h}=(n^{\alpha_1})_{i\in[k]}} = \thetao{\prod_{(v_i,v_j)\in E_H:\alpha_{i}+\alpha_{j}<1}n^{\alpha_{i}+\alpha_{j}-1}}.
\end{aligned}
\end{equation}
The hidden variables are an i.i.d.\ sample from a power-law distribution. Therefore,
\begin{equation}\label{eq:pdegreerange}
\begin{aligned}[b]
\Prob{h_i\in[\varepsilon,1/\varepsilon] (\mu n)^{\alpha}}& =\int_{\varepsilon (\mu n)^\alpha}^{1/\varepsilon(\mu n)^\alpha}cx^{-\tau}\dd x= K(\varepsilon)(\mu n)^{\alpha(1-\tau)}
\end{aligned}
\end{equation}
for some constant $K(\varepsilon)$ not depending on $n$. 
The expected number of vertices with degrees in $[\varepsilon,1/\varepsilon](\mu n)^\alpha$ scales as $\Theta(n^{(1-\tau)\alpha+1})$. 
Then, the number of sets of vertices with degrees in $[\varepsilon,1/\varepsilon](n^{\alpha_i})_{i\in [k]}$ scales as
\begin{equation}\label{eq:numhdeg}
\Theta\Big( n^{k+(1-\tau)\sum_i\alpha_i}\Big).
\end{equation}
Combining~\eqref{eq:phsub} and~\eqref{eq:numhdeg} yields that the contribution from vertices with degrees $n^{\boldsymbol{\alpha}}=(n^{\alpha_i})_{i\in[k]}$ to the expected number of motifs, $\Exp{N(H,\boldsymbol{\alpha}),\varepsilon}$ is
\begin{equation}\label{eq:Nalph}
\Exp{N(H,\boldsymbol{\alpha},\varepsilon)} = \thetao{ n^{k+(1-\tau)\sum_i\alpha_i} \ \ \prod_{\mathclap{(i,j)\in E_H:\alpha_i+\alpha_j<1}} \ \  n^{\alpha_1+\alpha_j-1}}.
\end{equation}
The maximum contribution is obtained for $\alpha_i$ that maximize
\begin{equation}\label{eq:maxalph}
\begin{aligned}[b]
&\max (1-\tau)\sum_{i}\alpha_i +\sum_{(i,j)\in E_H: \alpha_i+\alpha_j<1}\alpha_i+\alpha_j-1 
\end{aligned}
\end{equation} 
for $\alpha_i\geq 0$. The following lemma shows that this optimization problem attains its maximum for highly specific values of $\alpha$:

\begin{lemma}[Maximum contribution to expected number of motifs]\label{lem:maxmotif}
	Let $H$ be a connected graph on $k$ vertices. If the solution to~\eqref{eq:maxalph} is unique, then the optimal solution satisfies $\alpha_i\in\{0,\tfrac{1}{2},1\}$ for all $i$. If it is not unique, then there exist at least 2 optimal solutions with $\alpha_i\in\{0,\tfrac{1}{2},1\}$  for all $i$. 
\end{lemma}
\begin{proof}
	Defining $\beta_i=\alpha_i-\tfrac{1}{2}$ yields for~\eqref{eq:maxalph}
	\begin{equation}\label{eq:maxeqbeta}
	\max \frac{1-\tau}{2}k+ (1-\tau)\sum_{i}\beta_i +\sum_{(i,j)\in E_H: \beta_i+\beta_j<0}\beta_i+\beta_j,
	\end{equation} 
	over all possible values of $\beta_i\geq -\tfrac{1}{2}$. Then, we have to prove that $\beta_i\in\{-\tfrac 12, 0,\tfrac{1}{2}\}$ for all $i$ in the optimal solution.
	Note that~\eqref{eq:maxeqbeta} is a piecewise linear function in $\beta$. Therefore, if~\eqref{eq:maxeqbeta} has a unique maximum, it must be attained at the boundary for $\beta_i$ or at a border of one of the linear sections. Thus, any unique optimal value of $\beta_i$ satisfies $\beta_i=-\tfrac{1}{2}$ or $\beta_i+\beta_j=0$ for some $j$. Note that this implies that $\beta_i\leq\tfrac{1}{2}$.
	We ignore the constant factor of $(1-\tau)\tfrac{k}{2} $ in~\eqref{eq:maxeqbeta}, since it does not influence the optimal $\beta$ values.
	Rewriting~\eqref{eq:maxeqbeta} without the constant factor yields
	\begin{equation}\label{eq:maxbeta}
	\max \sum_i \beta_i\left(1-\tau +\# \{\text{ edges to }j\text{ with }\beta_j<-\beta_i\}\right).
	\end{equation}
	The proof of the lemma then consists of two steps: \\
	\textit{Step 1.} Show that any unique solution does not have vertices $i$ with $\abs{\beta_i}\in(0,\tfrac{1}{2})$.\\
	\textit{Step 2.} Show that any optimal solution that is not unique can be transformed into two different optimal solutions with $\beta_i\in\{-\tfrac 12, 0,\tfrac{1}{2}\}$ for all $i$.\\

	\textit{Step 1.}
	We show that when the solution to~\eqref{eq:maxbeta} is unique, it is never optimal to have $\abs{\beta}\in(0,\tfrac{1}{2})$. 
	Let 
	\begin{equation}\label{eq:tildebeta}
	\tilde{\beta}=\min_{i:\abs{\beta_i}>0}\abs{\beta_i}.
	\end{equation}
	If $\tilde{\beta}=\tfrac{1}{2}$, then we are finished, so assume that $\tilde{\beta}<\tfrac{1}{2}$.  Then, there exist $N_{\tilde{\beta}^-}$ vertices with their $\beta$ value equal to $-\tilde{\beta}$, and $N_{\tilde{\beta}^+}$ vertices with value $\tilde{\beta}$, where $N_{\tilde{\beta}^+}+N_{\tilde{\beta}^-}\geq 1$. Furthermore, let $E_{\tilde{\beta}^-}$ denote the number of edges from vertices with value $-\tilde{\beta}$ to other vertices $j$ such that $\beta_j<\tilde{\beta}$, and $E_{\tilde{\beta}^+}$ the number of edges from vertices with value $\tilde{\beta}$ to other vertices $j$ such that $\beta_j<-\tilde{\beta}$. Then, the contribution from these vertices to~\eqref{eq:maxbeta} is
	\begin{equation}\label{eq:Nbeta}
	\tilde{\beta}\left((1-\tau)\left(N_{\tilde{\beta}^+}-N_{\tilde{\beta}^-}\right)+E_{\tilde{\beta}^+}-E_{\tilde{\beta}^-}\right).
	\end{equation}
	Because the optimizer of~\eqref{eq:maxbeta} is unique, the term inside the brackets cannot equal zero. Then, increasing $\tilde{\beta}$ if the term inside brackets is larger than zero, or decreasing it if it is smaller than zero would improve the optimal contribution, which is not possible. 
	Thus, if the optimal solution is unique, we must have $\tilde{\beta}=\tfrac{1}{2}$. This shows that any unique solution satisfies $\beta_i\in\{-\tfrac{1}{2},0,\tfrac{1}{2}\}$ for all $i$.
	
	\textit{Step 2.}
	If the solution to~\eqref{eq:maxbeta} is not unique, then by the same argument that leads to~\eqref{eq:Nbeta}, there exist $\hat{\beta}_1,\ldots,\hat{\beta}_s>0$ for some $s\geq 1$ such that  
	\begin{equation}
	\hat{\beta}_j\left((1-\tau)\left(N_{\hat{\beta}_j^+}-N_{\hat{\beta}_j^-}\right)+E_{\hat{\beta}_j^+}-E_{\hat{\beta}_j^-}\right)=0 \quad \forall j\in[s].
	\end{equation}
	Here we use the same notation as in~\eqref{eq:Nbeta}.
	All other values of $\beta$ must either be $0,\tfrac{1}{2}$ or $-\tfrac{1}{2}$ by the argument in Step 3. Thus, setting all $\hat{\beta}_j$ to zero does not change the value of the solution, and setting all $\hat{\beta}_j$ to $\tfrac{1}{2}$ also does not change the value of the solution. Thus, if the solution to~\eqref{eq:maxbeta} is not unique, at least 2 solutions exist with $\beta_i\in\{-\tfrac{1}{2},0,\tfrac{1}{2}\}$ for all $i$. 
\end{proof}

\paragraph*{Completion of the proof of Theorem~\ref{thm:expmotifs}.}
We first rewrite~\eqref{eq:maxalph} using Lemma~\ref{lem:maxmotif}. By Lemma~\ref{lem:maxmotif}, the maximal value of~\eqref{eq:maxalph} is attained by partitioning $V_H$ into the sets $S_1,S_2,S_3$ such that vertices in $S_1$ have $\alpha_i=0$, vertices in $S_2$ have $\alpha_i=1$ and vertices in $S_3$ have $\alpha_i=\tfrac{1}{2}$. Then, the edges with $\alpha_i+\alpha_j<1$ are edges inside $S_1$ and edges between $S_1$ and $S_3$. If we denote the number of edges inside $S_1$ by $E_{S_1}$ and the number of edges between $S_1$ and $S_3$ by $E_{S_1,S_3}$, then we can rewrite~\eqref{eq:maxalph} as
\begin{equation}\label{eq:maxtemp}
\begin{aligned}[b]
\max_{\mathcal{P}} \ (1-\tau)& \left(\abs{S_2}+\tfrac 12 \abs{S_3}\right)-E_{S_1}-\tfrac{1}{2}E_{S_1,S_3}
\end{aligned}
\end{equation}
over all partitions $\mathcal{P}$ of the vertices of $H$ into $S_1,S_2,S_3$. Using that $|S_3|=k-\abs{S_1}-\abs{S_2}$ yields
\begin{equation}\label{eq:maxtemp2}
\begin{aligned}[b]
\max_{\mathcal{P}} \ \frac{1-\tau}{2}k+\frac{\tau-1}{2}\left( \abs{S_1}-\abs{S_2}-\frac{2E_{S_1}+E_{S_1,S_3}}{\tau-1} \right),
\end{aligned}
\end{equation}
Since $k$ is fixed and $\tau-1>0$, maximizing~\eqref{eq:maxalph} is equivalent to maximizing
\begin{equation}\label{eq:maxeq}
B_f(H)=\max_{\mathcal{P}}\left[\abs{S_1}-\abs{S_2}-\frac{2E_{S_1}+E_{S_1,S_3}}{\tau-1}\right].
\end{equation}
Furthermore, by Lemma~\ref{lem:maxmotif}, the optimal value of~\eqref{eq:maxeq} is unique if and only if the solution to~\eqref{eq:maxalph} is unique. 

Let $\boldsymbol{\alpha}^*$ be the unique optimizer of~\eqref{eq:maxalph}. Then, by~\eqref{eq:Nalph}, for any $\boldsymbol{\alpha}\neq \boldsymbol{\alpha}^*$
\begin{equation}\label{eq:Nmaxcont}
\frac{\Exp{N(H,\boldsymbol{\alpha},\varepsilon)}}{\Exp{N(H,\boldsymbol{\alpha}^*,\varepsilon)}}=\thetao{n^{-\eta}}
\end{equation}
for some $\eta>0$. 
Combining this with~\eqref{eq:maxeq} proves the first part of the theorem. 
By~\eqref{eq:Nalph}, the contribution of the maximum is then given by 
\begin{equation}
\Exp{N(H,\boldsymbol{\alpha}^*,\varepsilon)}=n^kn^{\frac{1-\tau}{2}(k+B_f(H))}=	n^{\frac{3-\tau}{2}k+\frac{\tau-1}{2}B_f(H)},
\end{equation}
which proves the second part of the theorem. 

\section{Typical motif counts}\label{sec:motiftyp}
We now present equivalent versions of Lemma~\ref{lem:maxmotif} and Theorem~\ref{thm:expmotifs} for typical motif counts: 
\begin{lemma}[Maximum contribution to typical number of motifs]\label{lem:maxmotiftyp}
	Let $H$ be a connected graph on $k$ vertices. If the solution to~\eqref{eq:maxalph} with the extra constraint $\alpha_i\in[0,1/(\tau-1)]$ for all $i$ is unique, then the optimal solution satisfies $\alpha_i\in\{0,(\tau-2)/(\tau-1),\tfrac{1}{2},1/(\tau-1)\}$ for all $i$. Furthermore, $\alpha_i=0$ if and only if the degree of vertex $i$ in $H$ equals 1. If the solution is not unique, then there exist at least 2 optimal solutions with $\alpha_i\in\{0,(\tau-2)/(\tau-1),\tfrac{1}{2},1/(\tau-1)\}$  for all $i$. 
\end{lemma}

Thus, the optimization problem can be translated into the optimal partition $\mathcal{P}$ of all vertices of $H$ with degree at least 2 into three sets, $S_1, S_2$ and $S_3$. The vertices in $S_1$ correspond to vertices that have degree proportional to $n^{(\tau-2)/(\tau-1)}$ in the graph, vertices in $S_2$ correspond to the maximal degree vertices with degrees proportional to $n^{1/(\tau-1)}$, and $S_3$ corresponds to the vertices of $\sqrt{n}$ degrees. 

Let $E_{S_i,1}$ denote the number of edges between vertices in $S_i$ and vertices with degree 1 in $H$. Then, rewriting the optimal solution of~\eqref{eq:maxalph} results in
\begin{equation}\label{eq:maxtemptyp}
\begin{aligned}[b]
\max_{\mathcal{P}} \ & (1-\tau)\left(\frac{\tau-2}{\tau-1}\abs{S_1}+\frac{1}{\tau-1}\abs{S_2}+\tfrac 12 \abs{S_3}\right)\frac{\tau-3}{\tau-1}E_{S_1}+\frac{\tau-3}{2(\tau-1)}E_{S_1,S_3}-\frac{E_{S_1,1}}{\tau-1} \\
& -\frac{\tau-2}{\tau-1}E_{S_2,1}-\frac 12 E_{S_3,1},
\end{aligned}
\end{equation}
over all partitions $\mathcal{P}$ of the vertices of $H$ with degree at least 2 into $S_1,S_2,S_3$. Using that $|S_3|=k-\abs{S_1}-\abs{S_2}-k_1$, $E_{S_3,1}=k_1-E_{S_1,1}-E_{S_2,1}$, where $k_1$ denotes the number of degree 1 vertices in $H$, and extracting a factor $(3-\tau)/2$ results in
\begin{equation}\label{eq:maxtemptyp2}
\begin{aligned}[b]
\max_{\mathcal{P}} \ &  \frac{1-\tau}{2}k+\frac{(3-\tau)}{2}\Big( \abs{S_1}-\abs{S_2}+\frac{\tau-2}{3-\tau} k_1-\frac{2E_{S_1}+E_{S_1,S_3}}{\tau-1}  -\frac{E_{S_1,1}-E_{S_2,1}}{\tau-1}\Big),
\end{aligned}
\end{equation}
Since $k$ and $k_1$ are fixed and $3-\tau>0$, we need to maximize
\begin{equation}\label{eq:maxeqtyp}
B_t(H)=\max_{\mathcal{P}}\abs{S_1}-\abs{S_2}-\frac{2E_{S_1}+E_{S_1,S_3}+E_{S_1,1}-E_{S_2,1}}{\tau-1}.
\end{equation}

\begin{theorem}[General motifs, typical]\label{thm:expmotifstyp}
	Let $H$ be a motif on $k$ vertices such that the solution to~\eqref{eq:maxeqtyp} is unique and let $k_1$ denote the number of vertices of degree 1 in $H$. 
	With high probability, the number of motifs $H$ in a hidden-variable model  grows asymptotically as
	\begin{equation}
	N(H) \propto n^{\frac{3-\tau}{2}(k+B_t(H))+\frac{\tau-1}{2}k_1},
	\end{equation}
	and is thus fully determined by the partition $\mathcal{P}^*$ that optimizes \eqref{eq:maxeqtyp}. 
\end{theorem}

\begin{proof}
	The proof of Theorem~\ref{thm:expmotifstyp} as well as Lemma~\ref{lem:maxmotiftyp} follows the same lines as the proof of~\cite[Theorem 2.1 and Lemma 3.2]{hofstad2017d} for the erased configuration model instead of the hidden-variable model. That proof relies on the fact that the connection probability between vertices with degrees $D_i$ and $D_j$ in the erased configuration model has the same order of magnitude as $\min(D_iD_j/(\mu n),1)$, which is the same as the connection probability in the hidden-variable model. Thus, the proof of~\cite[Theorem 2.1]{hofstad2017d} also holds for the hidden-variable model.
\end{proof}

\section{Graphlets}\label{sec:graphlets}
We now focus on graphlet counting. Thus, we now only count the number of times the subgraph $H$ appears as an induced subgraph of a hidden-variable model. In particular, edges that are not present in $H$ are also required not to be present in the graph. Because the probability that two edges between vertices of high degree are present equals one (see~\eqref{eq:pij}), this puts a constraint on the number of vertices that typically have high degree. When we again let the degrees in the graphlet scale as $n^{\alpha_i}$, we see that the probability that an edge $(i,j)$ is present equals one as soon as $\alpha_1+\alpha_j>1$. Thus, for the expected number of graphlets, the optimization problem corresponding to~\eqref{eq:maxalph} becomes 
\begin{equation}\label{eq:maxalphgraphlet}
\begin{aligned}[b]
&\max (1-\tau)\sum_{i}\alpha_i +\sum_{(i,j)\in E_H: \alpha_i+\alpha_j<1}\alpha_i+\alpha_j-1 \\
&\text{s.t. }\alpha_i+\alpha_j \leq 1 \quad \forall (i,j)\notin E_H,
\end{aligned}
\end{equation} 
where $E_H$ denotes the edge set of $H$. Again, this optimization problem is maximized for $\alpha_i\in{0,\tfrac{1}{2},1}$, so that similarly, the optimization problem corresponding to~\eqref{eq:maxeq} including the extra constraint then becomes
\begin{equation}\label{eq:maxeqexpgraphlet}
\begin{aligned}[b]
B_{g,f}(H)=&\max_{\mathcal{P}}\left[\abs{S_1}-\abs{S_2}-\frac{2E_{S_1}+E_{S_1,S_3}}{\tau-1}\right],\\
&\text{s.t. }(u,v)\in E_H \quad \forall u\in S_2,v\in S_2\cup S_3.
\end{aligned}
\end{equation}
This optimization problem again finds the most likely degrees of vertices that together form the graphlet $H$. Vertices in $S_1$ have degrees proportional to a constant, vertices in $S_2$ have degrees proportional to $n$ and vertices in $S_3$ have degrees proportional to $\sqrt{n}$. Using this optimization problem, we are also able to find the scaling of the number of graphlets:
\begin{theorem}[General graphlets, expected]\label{thm:graphletsexp}
	Let $H$ be a graphlet on $k$ vertices such that the solution to~\eqref{eq:maxeqexpgraphlet} is unique. 
	The expected number of graphlets $N_g(H)$ in a hidden-variable model  grows asymptotically as
	\begin{equation}
	\Exp{N_g(H)}\propto n^{\frac{3-\tau}{2}(k+B_{g,f}(H))},
	\end{equation}
	and is thus fully determined by the partition $\mathcal{P}^*$ that optimizes \eqref{eq:maxeqexpgraphlet}. 
\end{theorem}

For the typical number of graphlets, we obtain a similar result, by adding the constraint $\alpha_i\leq\tfrac{1}{\tau-1}$ to~\eqref{eq:maxalphgraphlet}. This results in the optimization problem (like its equivalent version for motifs in~\eqref{eq:maxeqtyp})
\begin{equation}\label{eq:maxeqtypgraphlet}
\begin{aligned}[b]
B_{g,t}(H) =& \max_{\mathcal{P}}\Big[\abs{S_1}-\abs{S_2} -\frac{2E_{S_1}+E_{S_1,S_3}+E_{S_1,1}-E_{S_2,1}}{\tau-1}\Big],\\
& \text{s.t. }(u,v)\in E_H \quad \forall u\in S_2,v\in S_2\cup S_3.
\end{aligned}
\end{equation}
The corresponding theorem to Theorem~\ref{thm:expmotifstyp} for graphlets then becomes:
\begin{theorem}[General motifs, typical]\label{thm:expgraphletstyp}
	Let $H$ be a graphlet on $k$ vertices such that the solution to~\eqref{eq:maxeqtypgraphlet} is unique and let $k_1$ denote the number of vertices of degree 1 in $H$. 
	With high probability, the number of graphlets $N_g(H)$ in a hidden-variable model  grows asymptotically as
	\begin{equation}
	N_g(H) \propto n^{\frac{3-\tau}{2}(k+B_{g,t}(H))+\frac{\tau-1}{2}k_1},
	\end{equation}
	and is thus fully determined by the partition $\mathcal{P}^*$ that optimizes \eqref{eq:maxeqtypgraphlet}. 
\end{theorem}


Figure~\ref{fig:graphlet5} shows the typical configuration of all graphlets on 5 vertices, similar to Fig.~\ref{fig:motif5} for all motifs on 5 vertices. Note that the typical configuration of a motif and a graphlet may be substantially different (for example in Fig.~\ref{fig:graphlet4types}(b) and its counterpart for motifs).

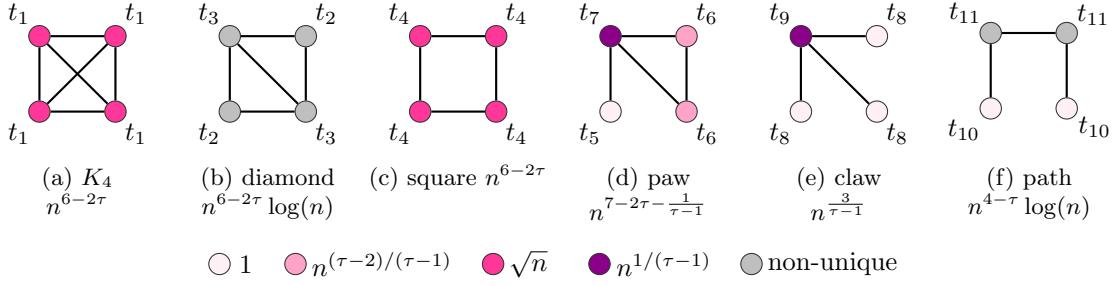
\begin{figure}[tb]
	\centering
	\captionsetup[subfigure]{justification=centering}
	\begin{subfigure}[t]{0.16\linewidth}
		\centering
		\begin{tikzpicture}
		\tikzstyle{edge} = [draw,thick,-]
		\node[S3,label={[label distance=0.05cm]230:$t_1$}] (a) at (0,0) {};
		\node[S3,label={[label distance=0.05cm]-50:$t_1$}] (b) at (1,0) {};
		\node[S3,label={[label distance=0.05cm]-230:$t_1$}] (c) at (0,1) {};
		\node[S3,label={[label distance=0.05cm]50:$t_1$}] (d) at (1,1) {};
		\draw[edge] (a)--(b);
		\draw[edge] (c)--(b);
		\draw[edge] (d)--(b);
		\draw[edge] (a)--(c);
		\draw[edge] (a)--(d);
		\draw[edge] (c)--(d);
		\end{tikzpicture}	
		\caption{$K_4$\quad \quad  $n^{6-2\tau}$}
		\label{fig:K4types}
	\end{subfigure}
	\begin{subfigure}[t]{0.16\linewidth}
		\centering
		\begin{tikzpicture}
		\tikzstyle{edge} = [draw,thick,-]
		\node[label={[label distance=0.05cm]230:$t_2$}] (a) at (0,0) {};
		\node[label={[label distance=0.05cm]-50:$t_3$}] (b) at (1,0) {};
		\node[label={[label distance=0.05cm]-230:$t_3$}] (c) at (0,1) {};
		\node[label={[label distance=0.05cm]50:$t_2$}] (d) at (1,1) {};
		\draw[edge] (a)--(b);
		\draw[edge] (c)--(b);
		\draw[edge] (d)--(b);
		\draw[edge] (a)--(c);
		\draw[edge] (c)--(d);
		\end{tikzpicture}	
		\caption{diamond $n^{6-2\tau}\log(n)$}
		\label{fig:squareextratype}
	\end{subfigure}
	\begin{subfigure}[t]{0.16\linewidth}
		\centering
		\begin{tikzpicture}
		\tikzstyle{edge} = [draw,thick,-]
		\node[S3,label={[label distance=0.05cm]230:$t_4$}] (a) at (0,0) {};
		\node[S3,label={[label distance=0.05cm]-50:$t_4$}] (b) at (1,0) {};
		\node[S3,label={[label distance=0.05cm]-230:$t_4$}] (c) at (0,1) {};
		\node[S3,label={[label distance=0.05cm]50:$t_4$}] (d) at (1,1) {};
		\draw[edge] (a)--(b);
		\draw[edge] (d)--(b);
		\draw[edge] (a)--(c);
		\draw[edge] (c)--(d);
		\end{tikzpicture}	
		\caption{square $n^{6-2\tau}$}
		\label{fig:squaretype}
	\end{subfigure}
	\begin{subfigure}[t]{0.16\linewidth}
		\centering
		\begin{tikzpicture}
		\tikzstyle{edge} = [draw,thick,-]
		\node[n1,label={[label distance=0.05cm]230:$t_5$}] (a) at (0,0) {};
		\node[S1,label={[label distance=0.05cm]-50:$t_6$}] (b) at (1,0) {};
		\node[S2m,label={[label distance=0.05cm]-230:$t_7$}] (c) at (0,1) {};
		\node[S1,label={[label distance=0.05cm]50:$t_6$}] (d) at (1,1) {};
		\draw[edge] (c)--(b);
		\draw[edge] (d)--(b);
		\draw[edge] (a)--(c);
		\draw[edge] (c)--(d);
		\end{tikzpicture}	
		\caption{paw $n^{7-2\tau-\frac{1}{\tau-1}}$}
		\label{fig:pawtype}
	\end{subfigure}
	\begin{subfigure}[t]{0.16\linewidth}
		\centering
		\begin{tikzpicture}
		\tikzstyle{edge} = [draw,thick,-]
		\node[n1,label={[label distance=0.05cm]230:$t_8$}] (a) at (0,0) {};
		\node[n1,label={[label distance=0.05cm]-50:$t_8$}] (b) at (1,0) {};
		\node[S2m,label={[label distance=0.05cm]-230:$t_9$}] (c) at (0,1) {};
		\node[n1,label={[label distance=0.05cm]50:$t_8$}] (d) at (1,1) {};
		\draw[edge] (c)--(b);
		\draw[edge] (a)--(c);
		\draw[edge] (c)--(d);
		\end{tikzpicture}	
		\caption{claw \quad\quad  $n^{\frac{3}{\tau-1}}$}
		\label{fig:wedge4type}
	\end{subfigure}
	\begin{subfigure}[t]{0.16\linewidth}
		\centering
		\begin{tikzpicture}
		\tikzstyle{edge} = [draw,thick,-]
		\node[n1,label={[label distance=0.05cm]210:$t_{10}$}] (a) at (0,0) {};
		\node[n1,label={[label distance=0.05cm]-20:$t_{10}$}] (b) at (1,0) {};
		\node[label={[label distance=0.03cm]-210:$t_{11}$}] (c) at (0,1) {};
		\node[label={[label distance=0.03cm]20:$t_{11}$}] (d) at (1,1) {};
		\draw[edge] (d)--(b);
		\draw[edge] (a)--(c);
		\draw[edge] (c)--(d);
		\end{tikzpicture}	
		\caption{path $n^{4-\tau}\log(n)$}
	\end{subfigure}
	
	\vspace{-0.5cm}
	\begin{subfigure}{\linewidth}
		\centering
		\begin{tikzpicture}
		\node[S2m,label={[label distance=0.05cm]0:$n^{1/(\tau-1)}$}] (a) at (5,0) {};
		\node[S3,label={[label distance=0cm]0:$\sqrt{n}$}] (b) at (3.6,0) {};
		\node[S1,label={[label distance=0.05cm]0:$n^{(\tau-2)/(\tau-1)}$}] (c) at (1,0) {};
		\node[n1,label={[label distance=0.05cm]0:$1$}] (c) at (0,0) {};
		\node[label={[label distance=0.05cm]0:non-unique}] (d) at (7,0) {};
		\end{tikzpicture}
	\end{subfigure}
	\vspace{-0.8cm}
	\caption{Scaling of the typical number of graphlets on 4 vertices in $n$. The vertex color indicates the typical vertex degree. The vertex labels indicate the vertex types used in Figure~\ref{fig:motifsdegree}}
	\label{fig:graphlet4types}
\end{figure}

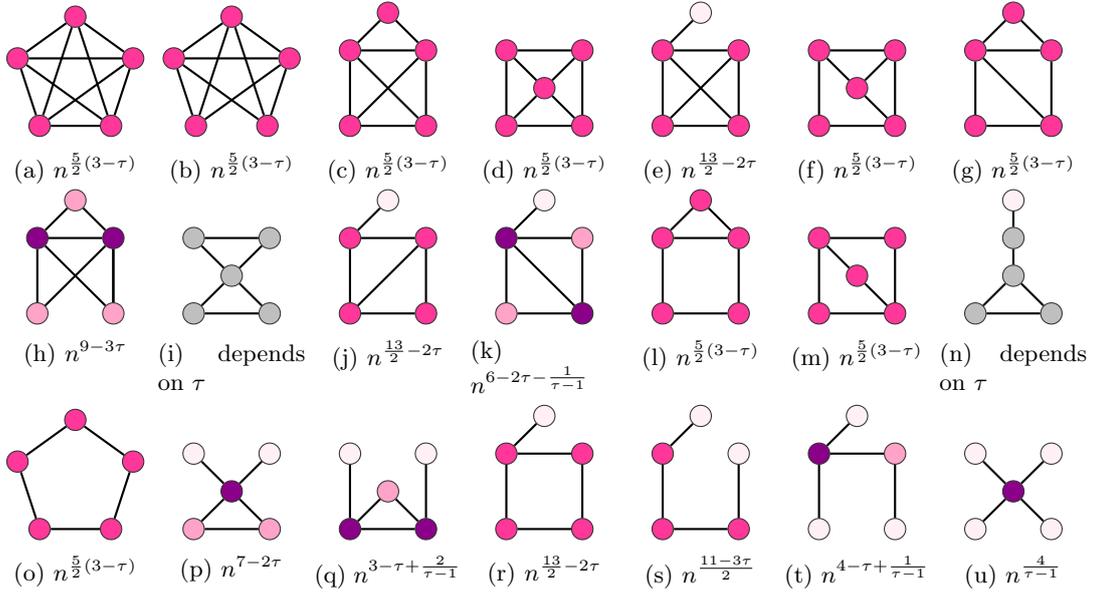
\begin{figure}[htb]
	\centering
	\begin{subfigure}[t]{0.13\linewidth}
		\centering
		\begin{tikzpicture}
		\tikzstyle{edge} = [draw,thick,-]
		\node[S3] (a) at (90:0.8) {};
		\node[S3] (b) at (162:0.8) {};
		\node[S3] (c) at (234:0.8) {};
		\node[S3] (d) at (306:0.8) {};
		\node[S3] (e) at (378:0.8) {};
		\draw[edge] (a)--(b);
		\draw[edge] (c)--(b);
		\draw[edge] (d)--(b);
		\draw[edge] (a)--(c);
		\draw[edge] (a)--(d);
		\draw[edge] (c)--(e);
		\draw[edge] (e)--(d);
		\draw[edge] (b)--(e);
		\draw[edge] (a)--(e);
		\draw[edge] (c)--(d);
		\end{tikzpicture}	
		\caption{$n^{\frac{5}{2}(3-\tau)}$}
	\end{subfigure}
	\begin{subfigure}[t]{0.13\linewidth}
		\centering
		\begin{tikzpicture}
		\tikzstyle{edge} = [draw,thick,-]
		\node[S3] (a) at (90:0.8) {};
		\node[S3] (b) at (162:0.8) {};
		\node[S3] (c) at (234:0.8) {};
		\node[S3] (d) at (306:0.8) {};
		\node[S3] (e) at (378:0.8) {};
		\draw[edge] (a)--(b);
		\draw[edge] (c)--(b);
		\draw[edge] (d)--(b);
		\draw[edge] (a)--(c);
		\draw[edge] (a)--(d);
		\draw[edge] (c)--(e);
		\draw[edge] (e)--(d);
		\draw[edge] (b)--(e);
		\draw[edge] (a)--(e);
		\end{tikzpicture}	
		\caption{$n^{\frac{5}{2}(3-\tau)}$}
	\end{subfigure}
	\begin{subfigure}[t]{0.13\linewidth}
		\centering
		\begin{tikzpicture}
		\tikzstyle{edge} = [draw,thick,-]
		\node[S3] (a) at (0,0) {};
		\node[S3] (b) at (1,0) {};
		\node[S3] (c) at (0,1) {};
		\node[S3] (d) at (1,1) {};
		\node[S3] (e) at (0.5,1.5) {};
		\draw[edge] (a)--(b);
		\draw[edge] (c)--(b);
		\draw[edge] (d)--(b);
		\draw[edge] (a)--(c);
		\draw[edge] (a)--(d);
		\draw[edge] (c)--(e);
		\draw[edge] (d)--(e);
		\draw[edge] (c)--(d);
		\end{tikzpicture}	
		\caption{$n^{\frac{5}{2}(3-\tau)}$}
	\end{subfigure}
	\begin{subfigure}[t]{0.13\linewidth}
		\centering
		\begin{tikzpicture}
		\tikzstyle{edge} = [draw,thick,-]
		\node[S3] (a) at (0,0) {};
		\node[S3] (b) at (1,0) {};
		\node[S3] (c) at (0,1) {};
		\node[S3] (d) at (1,1) {};
		\node[S3] (e) at (0.5,0.5) {};
		\draw[edge] (a)--(b);
		\draw[edge] (c)--(a);
		\draw[edge] (b)--(e);
		\draw[edge] (a)--(e);
		\draw[edge] (b)--(d);
		\draw[edge] (c)--(e);
		\draw[edge] (d)--(e);
		\draw[edge] (c)--(d);
		\end{tikzpicture}	
		\caption{$n^{\frac{5}{2}(3-\tau)}$}
	\end{subfigure}
	\begin{subfigure}[t]{0.13\linewidth}
		\centering
		\begin{tikzpicture}
		\tikzstyle{edge} = [draw,thick,-]
		\node[S3] (a) at (0,0) {};
		\node[S3] (b) at (1,0) {};
		\node[S3] (c) at (0,1) {};
		\node[S3] (d) at (1,1) {};
		\node[n1] (e) at (0.5,1.5) {};
		\draw[edge] (a)--(b);
		\draw[edge] (c)--(b);
		\draw[edge] (d)--(b);
		\draw[edge] (a)--(c);
		\draw[edge] (a)--(d);
		\draw[edge] (c)--(e);
		\draw[edge] (c)--(d);
		\end{tikzpicture}	
		\caption{$n^{\frac{13}{2}-2\tau}$}
	\end{subfigure}
	\begin{subfigure}[t]{0.13\linewidth}
		\centering
		\begin{tikzpicture}
		\tikzstyle{edge} = [draw,thick,-]
		\node[S3] (a) at (0,0) {};
		\node[S3] (b) at (1,0) {};
		\node[S3] (c) at (0,1) {};
		\node[S3] (d) at (1,1) {};
		\node[S3] (e) at (0.5,0.5) {};
		\draw[edge] (a)--(b);
		\draw[edge] (c)--(a);
		\draw[edge] (b)--(e);
		\draw[edge] (b)--(d);
		\draw[edge] (c)--(e);
		\draw[edge] (c)--(d);
		\draw[edge] (e)--(d);
		\end{tikzpicture}	
		\caption{$n^{\frac{5}{2}(3-\tau)}$}
	\end{subfigure}	
	\begin{subfigure}[t]{0.13\linewidth}
		\centering
		\begin{tikzpicture}
		\tikzstyle{edge} = [draw,thick,-]
		\node[S3] (a) at (0,0) {};
		\node[S3] (b) at (1,0) {};
		\node[S3] (c) at (0,1) {};
		\node[S3] (d) at (1,1) {};
		\node[S3] (e) at (0.5,1.5) {};
		\draw[edge] (a)--(b);
		\draw[edge] (c)--(b);
		\draw[edge] (d)--(b);
		\draw[edge] (a)--(c);
		\draw[edge] (c)--(e);
		\draw[edge] (d)--(e);
		\draw[edge] (c)--(d);
		\end{tikzpicture}	
		\caption{$n^{\frac{5}{2}(3-\tau)}$}
	\end{subfigure}
	
	\begin{subfigure}[t]{0.13\linewidth}
		\centering
		\begin{tikzpicture}
		\tikzstyle{edge} = [draw,thick,-]
		\node[S1] (a) at (0,0) {};
		\node[S1] (b) at (1,0) {};
		\node[ S2m] (c) at (0,1) {};
		\node[ S2m] (d) at (1,1) {};
		\node[S1] (e) at (0.5,1.5) {};
		\draw[edge] (d)--(b);
		\draw[edge] (c)--(b);
		\draw[edge] (d)--(b);
		\draw[edge] (a)--(c);
		\draw[edge] (a)--(d);
		\draw[edge] (c)--(e);
		\draw[edge] (d)--(e);
		\draw[edge] (c)--(d);
		\end{tikzpicture}	
		\caption{$n^{9-3\tau}$}
		\label{mot:fig:K23}
	\end{subfigure}
	\begin{subfigure}[t]{0.13\linewidth}
		\centering
		\begin{tikzpicture}
		\tikzstyle{edge} = [draw,thick,-]
		\node (a) at (0,0) {};
		\node (b) at (1,0) {};
		\node (c) at (0,1) {};
		\node (d) at (1,1) {};
		\node (e) at (0.5,0.5) {};
		\draw[edge] (a)--(b);
		\draw[edge] (e)--(a);
		\draw[edge] (b)--(e);
		\draw[edge] (e)--(d);
		\draw[edge] (c)--(e);
		\draw[edge] (c)--(d);
		\end{tikzpicture}	
		\caption{depends on $\tau$}
		\label{mot:fig:bowtie}
	\end{subfigure}
	\begin{subfigure}[t]{0.13\linewidth}
		\centering
		\begin{tikzpicture}
		\tikzstyle{edge} = [draw,thick,-]
		\node[S3] (a) at (0,0) {};
		\node[S3] (b) at (1,0) {};
		\node[S3] (c) at (0,1) {};
		\node[S3] (d) at (1,1) {};
		\node[n1] (e) at (0.5,1.5) {};
		\draw[edge] (a)--(b);
		\draw[edge] (d)--(b);
		\draw[edge] (a)--(c);
		\draw[edge] (a)--(d);
		\draw[edge] (c)--(e);
		\draw[edge] (c)--(d);
		\end{tikzpicture}	
		\caption{$n^{\frac{13}{2}-2\tau}$}
	\end{subfigure}
	\begin{subfigure}[t]{0.13\linewidth}
		\centering
		\begin{tikzpicture}
		\tikzstyle{edge} = [draw,thick,-]
		\node[S1] (a) at (0,0) {};
		\node[ S2m] (b) at (1,0) {};
		\node[ S2m] (c) at (0,1) {};
		\node[S1] (d) at (1,1) {};
		\node[n1] (e) at (0.5,1.5) {};
		\draw[edge] (a)--(b);
		\draw[edge] (d)--(b);
		\draw[edge] (a)--(c);
		\draw[edge] (b)--(c);
		\draw[edge] (c)--(e);
		\draw[edge] (c)--(d);
		\end{tikzpicture}	
		\caption{$n^{6-2\tau-\frac{1}{\tau-1}}$}
	\end{subfigure}
	\begin{subfigure}[t]{0.13\linewidth}
		\centering
		\begin{tikzpicture}
		\tikzstyle{edge} = [draw,thick,-]
		\node[S3] (a) at (0,0) {};
		\node[S3] (b) at (1,0) {};
		\node[S3] (c) at (0,1) {};
		\node[S3] (d) at (1,1) {};
		\node[S3] (e) at (0.5,1.5) {};
		\draw[edge] (a)--(b);
		\draw[edge] (d)--(b);
		\draw[edge] (a)--(c);
		\draw[edge] (e)--(d);
		\draw[edge] (c)--(e);
		\draw[edge] (c)--(d);
		\end{tikzpicture}	
		\caption{$n^{\frac{5}{2}(3-\tau)}$}
	\end{subfigure}	
	\begin{subfigure}[t]{0.13\linewidth}
		\centering
		\begin{tikzpicture}
		\tikzstyle{edge} = [draw,thick,-]
		\node[S3] (a) at (0,0) {};
		\node[S3] (b) at (1,0) {};
		\node[S3] (c) at (0,1) {};
		\node[S3] (d) at (1,1) {};
		\node[S3] (e) at (0.5,0.5) {};
		\draw[edge] (a)--(b);
		\draw[edge] (c)--(a);
		\draw[edge] (b)--(e);
		\draw[edge] (b)--(d);
		\draw[edge] (c)--(e);
		\draw[edge] (c)--(d);
		\end{tikzpicture}	
		\caption{$n^{\frac{5}{2}(3-\tau)}$}
		\label{mot:fig:m5dom}
	\end{subfigure}
	\begin{subfigure}[t]{0.13\linewidth}
		\centering
		\begin{tikzpicture}
		\tikzstyle{edge} = [draw,thick,-]
		\node (a) at (0,0) {};
		\node (b) at (1,0) {};
		\node (c) at (0.5,1) {};
		\node[n1] (d) at (0.5,1.5) {};
		\node (e) at (0.5,0.5) {};
		\draw[edge] (a)--(b);
		\draw[edge] (e)--(a);
		\draw[edge] (b)--(e);
		\draw[edge] (e)--(c);
		\draw[edge] (c)--(d);
		\end{tikzpicture}	
		\caption{depends on $\tau$}
		\label{mot:fig:dependtau}
	\end{subfigure}
	
	\begin{subfigure}[t]{0.13\linewidth}
		\centering
		\begin{tikzpicture}
		\tikzstyle{edge} = [draw,thick,-]
		\node[S3] (a) at (90:0.8) {};
		\node[S3] (b) at (162:0.8) {};
		\node[S3] (c) at (234:0.8) {};
		\node[S3] (d) at (306:0.8) {};
		\node[S3] (e) at (378:0.8) {};
		\draw[edge] (a)--(b);
		\draw[edge] (c)--(b);
		\draw[edge] (d)--(c);
		\draw[edge] (a)--(e);
		\draw[edge] (d)--(e);
		\end{tikzpicture}	
		\caption{$n^{\frac{5}{2}(3-\tau)}$}
	\end{subfigure}
	\begin{subfigure}[t]{0.13\linewidth}
		\centering
		\begin{tikzpicture}
		\tikzstyle{edge} = [draw,thick,-]
		\node[S1] (a) at (0,0) {};
		\node[S1] (b) at (1,0) {};
		\node[n1] (c) at (0,1) {};
		\node[n1] (d) at (1,1) {};
		\node[ S2m] (e) at (0.5,0.5) {};
		\draw[edge] (a)--(b);
		\draw[edge] (e)--(a);
		\draw[edge] (b)--(e);
		\draw[edge] (e)--(d);
		\draw[edge] (c)--(e);
		\end{tikzpicture}	
		\caption{$n^{7-2\tau}$}
	\end{subfigure}
	\begin{subfigure}[t]{0.13\linewidth}
		\centering
		\begin{tikzpicture}
		\tikzstyle{edge} = [draw,thick,-]
		\node[ S2m] (a) at (0,0) {};
		\node[ S2m] (b) at (1,0) {};
		\node[n1] (c) at (0,1) {};
		\node[n1] (d) at (1,1) {};
		\node[S1] (e) at (0.5,0.5) {};
		\draw[edge] (a)--(b);
		\draw[edge] (e)--(a);
		\draw[edge] (b)--(e);
		\draw[edge] (b)--(d);
		\draw[edge] (c)--(a);
		\end{tikzpicture}	
		\caption{$n^{3-\tau+\frac{2}{\tau-1}}$}
	\end{subfigure}
	\begin{subfigure}[t]{0.13\linewidth}
		\centering
		\begin{tikzpicture}
		\tikzstyle{edge} = [draw,thick,-]
		\node[S3] (a) at (0,0) {};
		\node[S3] (b) at (1,0) {};
		\node[S3] (c) at (0,1) {};
		\node[S3] (d) at (1,1) {};
		\node[n1] (e) at (0.5,1.5) {};
		\draw[edge] (a)--(b);
		\draw[edge] (d)--(b);
		\draw[edge] (a)--(c);
		\draw[edge] (c)--(e);
		\draw[edge] (c)--(d);
		\end{tikzpicture}	
		\caption{$n^{\frac{13}{2}-2\tau}$}
	\end{subfigure}
	\begin{subfigure}[t]{0.13\linewidth}
		\centering
		\begin{tikzpicture}
		\tikzstyle{edge} = [draw,thick,-]
		\node[S3] (a) at (0,0) {};
		\node[S3] (b) at (1,0) {};
		\node[S3] (c) at (0,1) {};
		\node[n1] (d) at (1,1) {};
		\node[n1] (e) at (0.5,1.5) {};
		\draw[edge] (a)--(b);
		\draw[edge] (d)--(b);
		\draw[edge] (a)--(c);
		\draw[edge] (c)--(e);
		\end{tikzpicture}	
		\caption{$n^{\frac{11-3\tau}{2}}$}
	\end{subfigure}
	\begin{subfigure}[t]{0.13\linewidth}
		\centering
		\begin{tikzpicture}
		\tikzstyle{edge} = [draw,thick,-]
		\node[n1] (a) at (0,0) {};
		\node[n1] (b) at (1,0) {};
		\node[ S2m] (c) at (0,1) {};
		\node[S1] (d) at (1,1) {};
		\node[n1] (e) at (0.5,1.5) {};
		\draw[edge] (d)--(b);
		\draw[edge] (a)--(c);
		\draw[edge] (c)--(e);
		\draw[edge] (c)--(d);
		\end{tikzpicture}	
		\caption{$n^{4-\tau+\frac{1}{\tau-1}}$}
	\end{subfigure}
	\begin{subfigure}[t]{0.13\linewidth}
		\centering
		\begin{tikzpicture}
		\tikzstyle{edge} = [draw,thick,-]
		\node[n1] (a) at (0,0) {};
		\node[n1] (b) at (1,0) {};
		\node[n1] (c) at (0,1) {};
		\node[n1] (d) at (1,1) {};
		\node[ S2m] (e) at (0.5,0.5) {};
		\draw[edge] (e)--(a);
		\draw[edge] (b)--(e);
		\draw[edge] (e)--(d);
		\draw[edge] (c)--(e);
		\end{tikzpicture}	
		\caption{$n^{\frac{4}{\tau-1}}$}
	\end{subfigure}
	\caption{Typical number of graphlets for all connected graphs on 5 vertices (constants ignored). The vertex color indicates the dominant vertex degree as in Figure~\ref{fig:motif5}.}
	\label{fig:graphlet5}
\end{figure}

\section{Proof of Theorem~\ref{thm:selfav}}\label{sec:selfavsqrt}
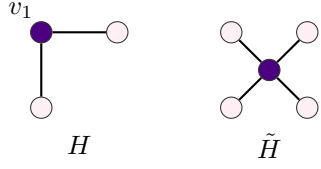
\begin{figure}[tb]
	\centering
	\begin{tikzpicture}
	\tikzstyle{edge} = [draw,thick,-]
	\node[n1] (a) at (0,0) {};
	\node[n1] (b) at (1,1) {};
	\node[S2,label={[label distance=0.05cm]-230:$v_1$}] (c) at (0,1) {};
	\draw[edge] (a)--(c);
	\draw[edge] (c)--(b);
	
	\node[draw=none,rectangle,fill=none] at (0.5,-0.5) {$H$};
	
	\begin{scope}[xshift=2.5cm]
	\node[n1] (a) at (0,0) {};
	\node[n1] (b) at (0,1) {};
	\node[n1] (d) at (1,0) {};
	\node[n1] (e) at (1,1) {};
	\node[S2] (c) at (0.5,0.5) {};
	\draw[edge] (a)--(c);
	\draw[edge] (c)--(b);
	\draw[edge] (d)--(c);
	\draw[edge] (c)--(e);
	
	\node[draw=none,rectangle,fill=none] at (0.5,-0.5) {$\tilde{H}$};
	\end{scope}
	\end{tikzpicture}	
	\caption{Construction of $\tilde{H}$ when $H$ is a path of length 3. $\tilde{H}$ is constructed by merging two copies of $H$ at vertex $v_1$ with most likely degree $n$. }
	\label{fig:3wedgeselfavg}
\end{figure}
Let $H$ be a motif on $k$ vertices such that the dominant motif structure of Theorem~\ref{thm:expmotifsx} contains non-$\sqrt{n}$ vertices. A similar analysis as the triangle example shows that the only motif on 2 vertices is non-self-averaging. Thus, we may assume that $k\geq 3$. Let the expected number of motifs of type $H$ as predicted in Theorem~\ref{thm:expmotifsx} be denoted by $n^{f(H)}$, and the optimal contribution from vertices $1,\dots, k$ in $H$ by $(n^{\alpha_i})_{i\in[k]}$. By Theorem~\ref{thm:expmotifsx},
\begin{equation}
\begin{aligned}[b]
\Exp{N(H)}& \propto \Exp{\# \text{ vertices having weights }(n^{\alpha_i})_{i\in[k]}} \Prob{H\text{ present on weights }(n^{\alpha_{i}})_{i\in[k]}}
\end{aligned}
\end{equation}
Suppose there is a vertex $v_1$ in $H$ with an optimal contribution of weight $n$ vertices, that is $\alpha_1=1$. Then, we study the contribution to the variance in~\eqref{eq:varnh} from the motif $\tilde{H}$ which is the motif on $2k-1$ vertices where two copies of motif $H$ are merged at their vertex $v_1$ (see Fig.~\ref{fig:3wedgeselfavg} for an example). We now investigate the expected number of $\tilde{H}$ motifs. In particular, we study the contribution to the expected number of $\tilde{H}$ motifs on vertices with weights of the order of magnitude $(n^{\beta_i})_{i\in[2k-1]}$ with
\begin{equation}\label{eq:alphbet}
\beta_i=\alpha_{t(i)},
\end{equation}
where $t(i)$ is the vertex in $H$ corresponding to vertex $i$ in $\tilde{H}$. Figure~\ref{fig:3wedgeselfavg} gives an example of this contribution when $H$ is a path with 3 vertices. Because of the construction of $\tilde{H}$,
\begin{equation}\label{eq:ptildeH}
\begin{aligned}[b]
&\Prob{\tilde{H}\text{ present on weights }(n^{\beta_i})_{i\in[2k-1]}} =\Prob{H\text{ present on weights }(n^{\alpha_i})_{i\in[k]}}^2.
\end{aligned}
\end{equation}
Furthermore, the only difference between the vertices in $\tilde{H}$ and two separate versions of $H$, is that $\tilde{H}$ contains one less vertex of weight proportional to $n$. The expected number of vertices with weight proportional to $n$ is given by $n^{2-\tau}$ Therefore,
\begin{equation}
\begin{aligned}[b]
&\Exp{\# \text{ vertices on weights }(n^{\beta_i})_{i\in[2k-1]}} \propto \frac{\Exp{\# \text{ vertices on weights }(n^{\alpha_i})_{i\in[k]}}^2 }{n^{2-\tau}}.
\end{aligned}
\end{equation}
Thus,
\begin{equation}
\begin{aligned}[b]
\expec[N(\tilde{H})]& \geq \Prob{H\text{ present on weights }(n^{\alpha_i})_{i\in[k]}}^2 \frac{\Exp{\# \text{ vertices on weights }(n^{\alpha_i})_{i\in[k]}}^2 }{n^{2-\tau}}\\
& \propto \Exp{N(H)}^2 n^{\tau-2}.
\end{aligned}
\end{equation}
Combining this with~\eqref{eq:varnh} of the main text results in
\begin{equation}
\frac{\Var{N(H)}}{\Exp{N(H)}^2}\geq \frac{\expec[N(\tilde{H})]}{\Exp{N(H)}^2}\geq n^{\tau-2},
\end{equation}
which diverges, because $\tau\in(2,3)$. Thus, if the optimal contribution to $H$ satisfies $S_2\neq\emptyset$, $H$ cannot be self-averaging for $\tau\in(2,3)$.

Now we study the case where $H$ has optimal contribution with $S_1\neq \emptyset$, but no vertices of weights proportional to $n$ so that $S_2=\emptyset$. Let $v\in S_1$. The contribution from $v$ to~\eqref{eq:maxeqsup} is
\begin{equation}
1-\frac{2d_{v,S_1}+d_{v,S_3}}{\tau-1},
\end{equation}
where $d_{v,S_i}$ denotes the number of edges from $v$ to vertices in set $S_i$. Moving $v$ to $S_3$ would change the contribution to
\begin{equation}
-\frac{d_{v,S_1}}{\tau-1}.
\end{equation}
Because we know that $v$ in $S_1$ is the optimal contribution, 
\begin{equation}
-\frac{d_{v,S_1}}{\tau-1}<1-\frac{2d_{v,S_1}+d_{v,S_3}}{\tau-1},
\end{equation}
or $d_{v,S_1}+d_{v,S_3}\leq \tau-1$, so that $d_{v,S_1}+d_{v,S_3}\in\{0,1\}$. Thus, every vertex in $S_1$ has at most 1 edge to other vertices in $S_1$ or vertices in $S_3$. Since we have assumed that $S_2=\emptyset$, and $k>2$, this means that all vertices in $S_1$ have degree 1 inside the motif, and are connected to a vertex in $S_3$. W.l.o.g. assume that $v_1$ is a vertex such that $v_1\in S_3$ and $v_1$ has at least one connection to a vertex in $S_1$. As in the previous proof, we consider $\tilde{H}$ constructed by merging two copies of $H$ at $v_1$, as illustrated in Figure~\ref{fig:degoneselfavg}. Define $(\alpha_i)_{i\in[k]}$ as the maximal contribution to $H$. We define
\begin{equation}
\beta_i=\begin{cases}
\alpha_{t(i)} & t(i)\neq 1,\\
1 & t(i)=1.
\end{cases}
\end{equation}
That is, we study the contribution where all vertices in $\tilde{H}$ except $v_1$ have the same weight as in their counterpart in $H$. The weight of $v_1$ is proportional to $n$, whereas the counterpart of $v_1$ in $H$ had weight proportional to $\sqrt{n}$, as in the illustration in Fig.~\ref{fig:degoneselfavg}. We again study the contribution to $\Exp{N(\tilde{H})}$ from vertices of weights $(n^{\beta_i})_{i\in[2k-1]}$. We now compare the probability that $\tilde{H}$ exists on vertices of weights $(n^{\beta_i})_{i\in[2k-1]}$ to the probability that two copies of $H$ exist on weights $(n^{\alpha_i})_{i\in[k]}$. The difference between these two probabilities is that vertex $v_1$ in $\tilde{H}$ has weight $n$ instead of $\sqrt{n}$ in $H$. In $H$, $v_1$ is connected to at least one vertex of weight proportional to 1. The probability of this connection to be present is proportional to $n^{-1/2}$. In $\tilde{H}$, $v_1$ has weight $n$, so that the probability that the corresponding connections occur in $\tilde{H}$ is proportional to 1. The connection probabilities on vertices not connected to $v_1$ do not change, so that 	
\begin{equation}
\begin{aligned}[b]
& \Prob{\tilde{H}\text{ present on weights }(n^{\beta_i})_{i\in[2k-1]}}\geq \frac{\Prob{H\text{ present on weights }(n^{\alpha_i})_{i\in[k]}}^2}{n^{-1}}.
\end{aligned}
\end{equation}

\begin{figure}[tb]
	
	\centering
	\begin{tikzpicture}
	\tikzstyle{edge} = [draw,thick,-]
	\node[S3] (a) at (0,0) {};
	\node[S3] (b) at (1,0) {};
	\node[S3] (c) at (0,1) {};
	\node[S3,label={[label distance=0.05cm]50:$v_1$}] (d) at (1,1) {};
	\node[n1] (e) at (0.5,1.5) {};
	\draw[edge] (a)--(b);
	\draw[edge] (c)--(b);
	\draw[edge] (d)--(b);
	\draw[edge] (a)--(c);
	\draw[edge] (a)--(d);
	\draw[edge] (d)--(e);
	\draw[edge] (c)--(d);
	
	\node[draw=none,rectangle,fill=none] at (0.5,-0.5) {$H$};
	
	\begin{scope}[xshift=2.5cm]
	\tikzstyle{edge} = [draw,thick,-]
	\node[S3] (a) at (0,0) {};
	\node[S3] (b) at (1,0) {};
	\node[S3] (c) at (0.5,1) {};
	\node[S2] (d) at (1.5,1) {};
	\node[n1] (e) at (1,1.5) {};
	\node[S3] (a2) at (2,0) {};
	\node[S3] (b2) at (3,0) {};
	\node[S3] (c2) at (2.5,1) {};
	\node[n1] (e2) at (2,1.5) {};
	\draw[edge] (a)--(b);
	\draw[edge] (c)--(b);
	\draw[edge] (d)--(b);
	\draw[edge] (a)--(c);
	\draw[edge] (a)--(d);
	\draw[edge] (d)--(e);
	\draw[edge] (c)--(d);
	
	\draw[edge] (a2)--(b2);
	\draw[edge] (c2)--(b2);
	\draw[edge] (d)--(b2);
	\draw[edge] (a2)--(c2);
	\draw[edge] (a2)--(d);
	\draw[edge] (d)--(e2);
	\draw[edge] (c2)--(d);
	
	\node[draw=none,rectangle,fill=none] at (1.5,-0.5) {$\tilde{H}$};
	\end{scope}
	\end{tikzpicture}	
	\caption{Example of the construction of $\tilde{H}$. $\tilde{H}$ is constructed by merging two copies of $H$ at vertex $v_1$: a vertex with most likely degree $\sqrt{n}$ that is connected to a vertex of degree 1. }
	\label{fig:degoneselfavg}
\end{figure}
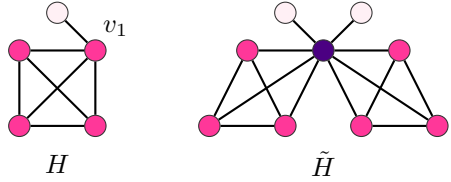

Since the difference between the vertices of two copies of $H$ and $\tilde{H}$ is that we remove two vertices of weight $\sqrt{n}$ and add one vertex of weight proportional to $n$, we obtain
\begin{equation}
\begin{aligned}[b]
& \Exp{\# \text{ vertices on weights }(n^{\beta_i})_{i\in[2k-1]}}  \propto \frac{\Exp{\# \text{ vertices on weights }(n^{\alpha_i})_{i\in[k]}}^2n^{2-\tau} }{n^{3-\tau}},
\end{aligned}
\end{equation}
where we have used that the number of vertices with degree proportional to $\sqrt{n}$ is $n^{(3-\tau)/2}$ by~\eqref{eq:numhdeg} and the number of vertices with degree proportional to $n$ scales as $n^{2-\tau}$. 
This results in 
\begin{equation}
\begin{aligned}[b]
\expec[N(\tilde{H})]& \geq \Prob{H\text{ present on degrees }(n^{\alpha_i})_{i\in[k]}}^2 \geq  \Exp{\# \text{ vertices on degrees }(n^{\alpha_i})_{i\in[k]}}^2\\
& \propto \Exp{N(H)}^2 
\end{aligned}
\end{equation}
so that by~\eqref{eq:varnh}
\begin{equation}
\frac{\Var{N(H)}}{\Exp{N(H)}^2}\geq \frac{\expec[N(\tilde{H})]}{\Exp{N(H)}^2}\propto 1,
\end{equation}
which does not converge to zero, so that the motif is not self-averaging.

\subsection*{Graphlet fluctuations}
The fluctuations of the number of graphlets can be studied similarly as the motif fluctuations. Again,~\eqref{eq:varnh} holds, but now it only includes graphlets $H_1,\ldots,H_l$ that can be constructed by merging two copies of $H$ at one or more vertices, with the additional constraint that after merging the copies of $H$, both copies still form induced subgraphs of $H$. As an example, consider 2-path fluctuations. The subgraphs that can be constructed from merging two 2-paths are the subgraphs in Fig.~\ref{fig:graphlet5}s, \ref{fig:graphlet5}t, \ref{fig:graphlet5}u, \ref{fig:graphlet4types}c, \ref{fig:graphlet4types}d, \ref{fig:graphlet4types}e, \ref{fig:graphlet4types}f, Fig. \ref{fig:graphlet4}a and~\ref{fig:graphlet4}b. However, the merged subgraph of Fig.~\ref{fig:graphlet4types}d and Fig.~\ref{fig:graphlet4types}b does not contain two induced copies of graphlets, and therefore these subgraphs are excluded from equation~\eqref{eq:varnh}. Other than that, the procedure to determine for any graphlet whether it is self-averaging or not is the same as the procedure for motifs, using Theorem~\ref{thm:graphletsexp} to find the expected order of magnitude of the number of merged graphlets.

\section{Optimal diamond motif}
We now study the optimizers of~\eqref{eq:maxeqx} for the diamond motif (Fig.~\ref{fig:graphlet4}b) for all $\tau\in(2,3)$. Let $i$ and $j$ be the vertices at the diagonal of the diamond, and $k$ and $s$ the corner vertices. 
By~\cite{stegehuiscode2019},~\eqref{eq:maxeqx} is maximized by  $S_3=\{i,j,k,s\}$ or $S_1=\{k,s\}$ and $S_2=\{i,j\}$ and by Lemma~\ref{lem:maxmotif} these are two optimizers of~\eqref{eq:maxeqx}, yielding a contribution of $2(1-\tau)$. 

By the argument in Lemma~\ref{lem:maxmotif}, for any vertex $u$, the optimizer of~\eqref{eq:maxeqx} is either given by $\alpha_u=0$ or constrained by $\alpha_u+\alpha_v=1$ for at least one edge $(u,v)\in E_H$. The above two optimizers show that 
for all vertices in the diamond motif, there is at least one such constrained edge. Therefore, $\alpha_i+\alpha_j+\alpha_k+\alpha_s=2$ must hold.


Then, to find all optimizers, we solve
\begin{align}
(1-\tau)(\alpha_i+\alpha_j+\alpha_k+\alpha_s)& +\ind{\alpha_i+\alpha_j<1}\left(\alpha_i+\alpha_j-1\right)+\ind{\alpha_i+\alpha_k<1}\left(\alpha_i+\alpha_k-1\right)\nonumber\\
& +\ind{\alpha_i+\alpha_s<1}\left(\alpha_i+\alpha_s-1\right) +\ind{\alpha_j+\alpha_s<1}\left(\alpha_j+\alpha_s-1\right)\nonumber\\
& +\ind{\alpha_j+\alpha_k<1}\left(\alpha_j+\alpha_k-1\right)=2(1-\tau).
\end{align} 
Using that $\alpha_i+\alpha_j+\alpha_k+\alpha_s=2$ this reduces to
\begin{align}
\alpha_i+\alpha_j+\alpha_k&+\alpha_s=2,\quad \alpha_i+\alpha_j\geq 1,\quad \alpha_i+\alpha_k\geq 1,\quad \alpha_i+\alpha_s\geq 1,\nonumber\\
&\alpha_j+\alpha_s\geq 1,\quad \alpha_j+\alpha_k\geq 1.
\end{align}
The constraints $ \alpha_i+\alpha_k\geq 1$, $\alpha_j+\alpha_s\geq 1$ together with $\alpha_i+\alpha_j+\alpha_k+\alpha_s=2$ yield $ \alpha_i+\alpha_k= 1$ and $\alpha_j+\alpha_s =1$. Similarly, $ \alpha_i+\alpha_s= 1$ and $\alpha_k+\alpha_j =1$, resulting in $\alpha_i=\alpha_j=:\beta$ and $\alpha_k=\alpha_s=1-\beta$. Furthermore, the constraint $\alpha_i+\alpha_j\geq 1$ yields $\beta\in[1/2,1]$.

\section{Graphlet counts}
\begin{table}[htbp]
	\centering
	\begin{tabular}{rrrrrr}
		\toprule
		\textbf{} & \textbf{Gowalla} & \textbf{Oregon} & \textbf{Enron} & \textbf{PGP} & \textbf{HEP} \\
		\midrule
		\textbf{claw} & 5.96$\cdot 10^{14}$ & 3.13$\cdot 10^{14}$  & 4.48$\cdot 10^{14}$  & 4.04$\cdot 10^{6}$ & 1.23$\cdot 10^{6}$ \\
		\textbf{path} & 1.52$\cdot 10^{14}$  & 7.03$\cdot 10^{7}$  & 1.37$\cdot 10^{14}$  & 2.72$\cdot 10^{6}$ & 2.12$\cdot 10^{6}$\\
		\textbf{paw} & 3.11$\cdot 10^{14}$ & 2.01$\cdot 10^{7}$ & 3.76$\cdot 10^{8}$  & 1.96$\cdot 10^{6}$ & 5.31$\cdot 10^{5}$ \\
		\textbf{diamond} & 8.60$\cdot 10^{7}$ & 9.61$\cdot 10^{5}$ & 2.25$\cdot 10^{7}$ & 2.74$\cdot 10^{5}$ & 3.55$\cdot 10^{4}$ \\
		\textbf{square} & 4.24$\cdot 10^{7}$ & 7.46$\cdot 10^{4}$ & 6.76$\cdot 10^{6}$ & 2.16$\cdot 10^{4}$ & 6.84$\cdot 10^{3}$ \\
		\textbf{K4} & 6.09$\cdot 10^{6}$ & 3.05$\cdot 10^{4}$ & 2.34$\cdot 10^{6}$& 2.39$\cdot 10^{5}$ & 6.56$\cdot 10^{4}$ \\
		\bottomrule
	\end{tabular}%
	\caption{Graphlet counts in the five data sets.}
	\label{tab:counts}%
\end{table}%

\end{document}